\newif\ifarxiv 
\arxivtrue

\documentclass[%
reprint,
superscriptaddress,
\ifarxiv nofootinbib,\else\fi
amsmath,amssymb,
aps,longbibliography
]{revtex4-2}

\usepackage[toc,page]{appendix} 

\usepackage{amsthm,bm,xparse,xcolor,tikz}
\usepgfmodule{decorations}
\usetikzlibrary{decorations,decorations.pathmorphing,patterns}
\usepackage[T1]{fontenc}

\usepackage[nodayofweek]{datetime}  
\usepackage{hyperref}
\usepackage{xcolor}
\usepackage{orcidlink}
\usepackage{soul}
\usepackage[capitalise]{cleveref}
\usepackage{enumerate}   
\definecolor{mylinkcolor}{rgb}{0,0,0.8} 
\hypersetup{unicode=true, %
  bookmarksnumbered=false,bookmarksopen=false,bookmarksopenlevel=1, %
  breaklinks=true,pdfborder={0 0 0},colorlinks=true}%
\hypersetup{%
  anchorcolor=mylinkcolor,citecolor=mylinkcolor, %
  filecolor=mylinkcolor,linkcolor=mylinkcolor, %
  menucolor=mylinkcolor,runcolor=mylinkcolor, %
  urlcolor=mylinkcolor}%

\newtheorem{lemma}{Lemma}

\newtheorem{proposition}{Proposition} 

\theoremstyle{definition}
\newtheorem{definition}{Definition}

\newcommand{\ket}[1]{| #1 \rangle}
\newcommand{\bra}[1]{\langle #1 |}

\newcommand{\ketbra}[2]{|#1\rangle\!\langle#2|}
\newcommand{\id}{\openone}

\usepackage[multiple]{footmisc}
\begin{document}

\title{Genuine multipartite entanglement is not necessary for standard device-independent conference key agreement}
\author{Lewis Wooltorton\hspace{0.5mm}\orcidlink{0000-0002-7950-3844}}
    \email{lewis.wooltorton@ens-lyon.fr}
    \affiliation{Inria, ENS de Lyon, LIP, 46 Allee d’Italie, 69364 Lyon Cedex 07, France}    \affiliation{Department of Mathematics, University of York, Heslington, York, YO10 5DD, United Kingdom}
    \affiliation{Quantum Engineering Centre for Doctoral Training, H. H. Wills Physics Laboratory and Department of Electrical \& Electronic Engineering, University of Bristol, Bristol BS8 1FD, United Kingdom}
\author{Peter Brown\hspace{0.5mm}\orcidlink{0000-0001-9593-0136}}
    \email{peter.brown@telecom-paris.fr}
    \affiliation{Télécom Paris, Inria, LTCI, Institut Polytechnique de Paris,
19 Place Marguerite Perey, 91120 Palaiseau, France}
\author{Roger Colbeck\hspace{0.5mm}\orcidlink{0000-0003-3591-0576}}
    \email{roger.colbeck@kcl.ac.uk}
\affiliation{Department of Mathematics, University of York, Heslington, York, YO10 5DD, United Kingdom}
    \affiliation{Department of Mathematics, King's College London, Strand, London, WC2R 2LS, United Kingdom}

\date{27$^{\text{th}}$ November, 2025}

\begin{abstract}
    Conference key agreement aims to establish shared, private randomness among many separated parties in a network. Device-independent conference key agreement (DICKA) is a variant in which the source and the measurement devices used by each party need not be trusted. So far, DICKA protocols largely fall into two categories: those that rely on violating a joint Bell inequality using genuinely multi-partite entangled states, and those that concatenate many bipartite protocols. The question of whether a hybrid protocol exists, where a multi-partite Bell inequality can be violated using only bipartite entanglement, was asked by Grasselli \emph{et al.}\ in [Quantum \textbf{7}, 980, (2023)]. We answer this question affirmatively, by constructing an asymptotically secure DICKA protocol achieving the same rate as the concatenation of bipartite device-independent quantum key distribution, yet relying on a single joint Bell violation. Our results prompt further discussion on the benefits of multi-partite entanglement for DICKA over its bipartite alternative, and we give an overview of different arguments for near-term devices.     
\end{abstract}
\maketitle

\ifarxiv\section{Introduction}\label{sec:intro}\else\noindent{\it Introduction.---}\fi
Information-theoretically secure communication between two separated parties requires a source of shared, private randomness, called a key. Conference key agreement (CKA) is an extension of key distribution to many parties, i.e., it aims to set up a key shared between $N>2$ parties which could then be used to secure communication amongst the $N$ parties. Quantum solutions exist~\cite{cabello00,Chen07,Epping_2017,Grasselli_2018, Grasselli_2019,MurtaReview20,Proietti_21,Carrara23,Pickston_2023} and, as in ordinary quantum key distribution (QKD), they can be made device-independent (DI)~\cite{Ribeiro18,Holz19,Ribeiro19,Holz20}. In other words, by witnessing certain nonlocal correlations~\cite{EPR,Bell_book,BarrettNonlocalResource}, one can show that a device-independent conference key agreement (DICKA) protocol is secure without any device characterization~\cite{BHK,ABGMPS,PABGMS,bcktwo,VV2}. 

A key question is to find the most efficient way to perform DICKA. Each network topology will determine which protocols can be performed, and a popular choice is the star network. Here, a central node distributes part of a multi-partite state to each party every round. One can then consider which states permit DICKA with this setup. Current literature has focused around the use of states with genuine multi-partite entanglement (GME), that is, states that are not biseparable\ifarxiv\else~\fi\footnote{A multi-partite state is biseparable if it can be written as a convex combination of states, where each state in the combination is separable across some bipartition.}, such as the Greenberger-Horne-Zeilinger (GHZ) state. In Refs.~\cite{Ribeiro18,Holz19,Ribeiro19} the parity Clauser-Horne-Shimony-Holt (CHSH) inequality was introduced as a multi-partite Bell inequality for DICKA maximally violated by the GHZ state. This was later generalized to another tailored inequality~\cite{Holz20}, and in Ref.~\cite{Grasselli_2023} analytical bounds on the von Neumann entropy, conditioned on witnessing their violation, were derived using techniques from Refs.~\cite{Ribeiro18,GrasselliMulti21,Woodhead_2018}. In Ref.~\cite{Grasselli_2023}, a comparison is also made to benchmark different DICKA protocols, including ones that are formed by concatenating bipartite device independent quantum key distribution (DIQKD), as described in Ref.~\cite{Epping_2017}. Upper bounds on DICKA rates were also studied in Refs.~\cite{Horodecki22,Philip23}. 

Many of the above protocols have the following structure, which we refer to as ``standard'': the central node distributes part of a multi-partite state to each party, and rounds are divided into test and generation. On test rounds, all parties test a multi-partite Bell expression, and on generation rounds they generate raw key. Reconciliation is performed where one party, say Alice, releases some error correction information that the other parties use to correct their strings. Privacy amplification then follows to ensure that an adversary has negligible information about the final key. The conference key rate is often defined as the difference between the entropy of Alice's raw string conditioned on the possible side information of an adversary and the highest cost of error correction between every other party and Alice, divided by the number of rounds~\cite{Grasselli_2023}.

Shared key can also be established using many independent bipartite protocols. Consider the same network, except that now the source distributes bipartite entanglement between Alice and every other party. Alice then engages in $N-1$ bipartite DIQKD protocols, and by the end holds a distilled key from each. She selects one of these as her final key, and publicly releases the exclusive OR (XOR) of the selected key with every other secret key she holds. Each party can then obtain the selected key using their private key. We refer to this type of protocol as a concatenation of bipartite DIQKD, and note that Alice should in principle use a separate device for each bipartite exchange to avoid opening memory loopholes~\cite{bckone} (or else make some additional assumption about her measurement device).

Concatenating bipartite DIQKD requires bipartite entanglement distribution only. However, it deviates from the standard protocol structure because it is based upon bipartite sub-protocols, whose security is proven independently. This raises the question asked in Ref.~\cite{Grasselli_2023}, of whether a standard realization of DICKA exists in which only biseparable states need to be distributed. An affirmative answer is known for device-dependent CKA~\cite{Carrara21}, where the following protocol structure was considered. Every round, rather than distributing $N-1$ bipartite entangled states, the source probabilistically chooses which party shares entanglement with Alice, and no trusted information about who was chosen is released. Then, the parties perform local measurements each round to generate raw key, and use standard QKD reconciliation. It is not {\it a priori} clear if non-zero conference key rate can be achieved with this protocol, and Ref.~\cite{Carrara21} showed it to be possible in the device-dependent setting. Specifically, Ref.~\cite{Carrara21} provided a family of biseparable states that lead to non-zero conference key in the $N$-BB84 protocol~\cite{Grasselli_2018}. 

Whilst biseparable states can be used for key agreement in the $N$-BB84 protocol, performing this in practice would be wasteful. If each party could learn when they share entanglement with Alice, which can be achieved with authenticated classical communication, the protocol reduces to a concatenation of bipartite QKD. This has better error reconciliation efficiency than that considered in Ref.~\cite{Carrara21} (see \ifarxiv\cref{sec:alt1} \else the Supplemental Material~\cite{supp} \fi for details), resulting in a higher conference key rate.

For the DI scenario, consider the case where each party is restricted to a single device, and a single multi-partite Bell expression is tested to derive security (i.e., the standard protocol structure is adopted). Can a positive conference key rate still be established with biseparable states? 

In this work, we answer the above question affirmatively. Biseparable states can be used for DICKA with a single device per party and a single Bell inequality. Our solution is operationally equivalent to the concatenation of many bipartite DIQKD protocols and hence more straightforward to implement with existing technology than one that requires GME. An additional advantage of our protocol is that it requires only one device per party, preventing particular memory attacks. Moreover, our results suggest the need for further discussion on the realistic benefits of multi-partite entanglement for DICKA over bipartite concatenation. 

\ifarxiv Our work is structured as follows. In \cref{sec:pro_1,sec:pro1} we describe our protocol. In \cref{sec:sec} we review common security definitions, and prove the protocol is secure in the asymptotic noiseless scenario, with details deferred to \cref{app:proof}. \cref{sec:alt} then contains a discussion of possible adaptations and variants of the protocol. Finally, in \cref{sec:disc} we discuss the existing justification for DICKA with GHZ states, and how this compares to protocols based on bipartite entanglement, such as the one presented here. \else\fi           
 
\ifarxiv\section{Standard DICKA protocol without GME}\label{sec:protocol}\else\medskip\noindent{\it Standard DICKA protocol without GME.---}\fi
\ifarxiv\subsection{Protocol overview}\label{sec:pro_1}\else\fi
\begin{figure}[h]
\includegraphics[width=8.4cm]{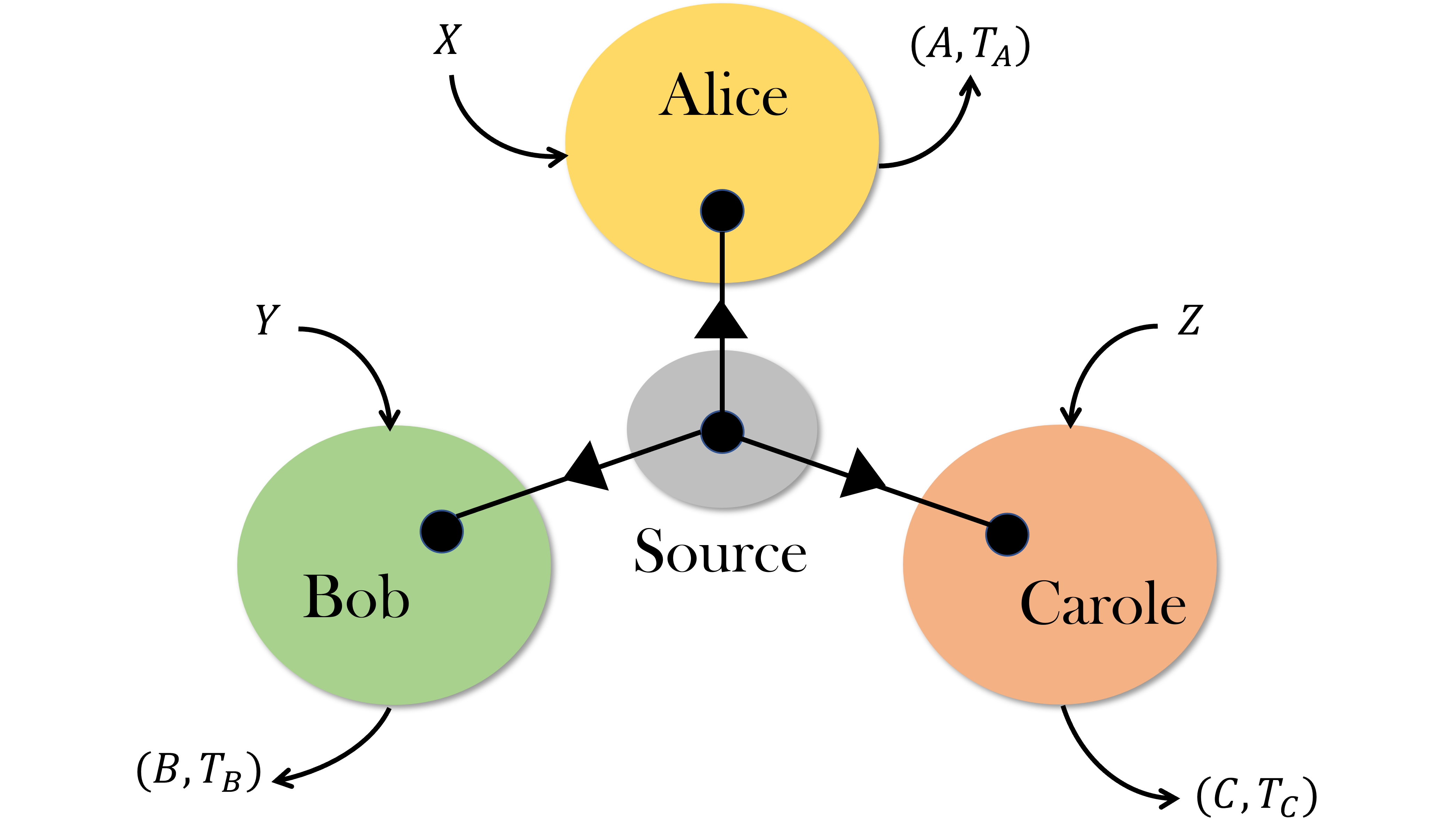}
\centering
\caption{Graphical representation of our DICKA protocol. Each round, every party receives part of a quantum state from an untrusted source, randomly samples an input $X/Y/Z$ and performs an untrusted measurement. They obtain two outcomes, and the joint statistics are used to estimate the value of a multi-partite Bell inequality. In the honest implementation, the outcome $T_{A/B/C}$ corresponds to a classical flag indicating who shares entanglement with Alice on that round, whilst the outcome $A/B/C$ results from performing a CHSH-type measurement on a quantum system.}
\label{fig:protocol}
\end{figure}
Our protocol operates in the following way: each round, a source distributes a multi-partite state, on which isolated measurements are made and a multi-partite Bell inequality is tested. Witnessing its maximum violation guarantees a conference key rate equal to that of $N-1$ concatenated bipartite DIQKD protocols. This can be achieved by a biseparable state that includes a flag register indicating who shares entanglement each round. 

Importantly, security is derived entirely from the maximum violation of the introduced Bell inequality, without trusting the flag. To achieve this, each user performs one side of a CHSH test~\cite{CHSH}. The average CHSH violation of each party with Alice is then measured, conditioned on the flag indicating said party shared entanglement; the sum of all such violations constitutes our Bell value. Any dishonesty in the flag will manifest as a reduction in the CHSH value for one of the pairings, ensuring security. Since every party has access to the flag, which must be reliable in the case of a high Bell violation, they can engage in some classical communication to each distill a secret key with Alice. Efficient reconciliation can then be conducted by Alice announcing the bitwise XOR of a chosen final key with every other secret key. An advantage of this protocol is that it is not vulnerable to the same memory attacks~\cite{bckone} as concatenating $N-1$ DIQKD protocols~\cite{Grasselli_2023}, implying that one conference key can still be established with a single device per party.

\ifarxiv\subsection{Technical description} \label{sec:pro1}\else\medskip\noindent{\it Technical description.---}\fi
For clarity, we focus on three parties, Alice, Bob and Carole, but the generalization to $n$ parties is straightforward. Assume the three parties are arranged in a star configuration, and each have access to one device, as shown in \cref{fig:protocol}. We use random variables $X$ for the input to Alice's device and $(A,T_A)$ for the outputs. These take values $x\in\{0,1\}$ and $(a,t_{A})\in \{0,1\}^{2}$, respectively\ifarxiv\else~\fi\footnote{We use upper case for random variables and lower case for particular instances.}. Bob and Carole's devices have three inputs and four outcomes, described by the random variables $Y,Z$ and $(B,T_{B}),\ (C,T_{C})$, which take values $y,z \in\{0,1,2\}$ and $(b,t_{B}),\ (c,t_{C})\in\{0,1\}^{2}$, respectively. We denote the final conference key $\mathsf{K}_{\mathrm{CKA}}$. Fixing this Bell scenario, we describe the following protocol\ifarxiv\else~\fi\footnote{As is standard in cryptography we use the following assumptions: 1.~Each party works within an isolated laboratory within which they can control information flow to the outside world as well as to individual devices within their laboratory. 2.~Each party has their own private random number generator. 3.~The eavesdropper is limited by the laws of physics and is computationally unbounded. 4.~The parties can communicate classically through authenticated public channels.}:  
\begin{enumerate}
    \item Each round, an unknown tripartite state is distributed among all parties and each party announces receipt of their part of the state.  Alice then randomly assigns the round as ``test'' or ``generation'', and communicates her choice to Bob and Carole.   
    \item On test rounds, each party uses their private random number generator to choose an input $x,y,z \in \{0,1\}$, and on generation rounds Alice sets $X=0$, and Bob and Carole set $Y=Z=2$.
    \item Each party performs their measurement and stores the outcomes in their respective classical registers\ifarxiv\else~\fi\footnote{Each device only learns its own input and not whether it is a test or generation round. See \ifarxiv\cref{app:info} \else the Supplemental Material~\cite{supp} \fi for details.}.  
    \item After a sufficient number of rounds, Alice, Bob and Carole announce the values of their registers $T_{A}$, $T_{B}$ and $T_{C}$ for all rounds. They abort if any of their strings differ or if the strings are constant.
    \item Each party publicly announces all of their inputs and outputs from the test rounds, and they check for a multi-partite Bell inequality violation. If the violation is less than some threshold\ifarxiv\else~\fi\footnote{For the purposes of this proof-of-principle demonstration, this threshold is set to the maximum Bell violation (cf.\ \ifarxiv\cref{prop:ent}\else Proposition~1\fi).}, they abort\ifarxiv\else~\fi\footnote{Additionally, all parties need to release a fraction of their generation data to perform an alignment test with Alice. If the fraction of matched outcomes is below a threshold, they will also abort.}. 
    \item Using the flag data, Alice and Carole communicate to distill a private key\ifarxiv\else~\fi\footnote{In the event of maximum Bell violation, the key distillation process is trivial, since all raw key bits are secure.} $\mathsf{K}_{AC}$. Similarly, Alice and Bob distill a private key $\mathsf{K}_{AB}$. 
    \item Alice sets $\mathsf{K}_{\mathrm{CKA}} = \mathsf{K}_{AB}$, and computes $\mathsf{K}_{\mathrm{XOR}} := \mathsf{K}_{AC} \oplus \mathsf{K}_{AB}$, which she announces. Carole computes $\mathsf{K}_{\mathrm{CKA}} = \mathsf{K}_{\mathrm{XOR}} \oplus \mathsf{K}_{AC}$, meanwhile Bob sets $\mathsf{K}_{\mathrm{CKA}} = \mathsf{K}_{AB}$.   
\end{enumerate}
In this protocol all classical communication takes place using authenticated public channels (i.e., an eavesdropper can overhear, but not modify these messages) and the quantum communication uses insecure quantum channels (i.e., an eavesdropper can do anything allowed by quantum physics to these signals). It is important that during the protocol untrusted devices do not learn anything other than what they need to perform the protocol, see \ifarxiv\cref{app:info} \else the Supplemental Material~\cite{supp} \fi for details.

We now consider an ideal honest implementation using a biseparable state. In this implementation, the state shared between the honest parties is
\begin{multline}
    \rho =  \frac{1}{2}\ketbra{\Phi_{0}}{\Phi_{0}}_{Q_{A}Q_{B}} \otimes \ketbra{+}{+}_{Q_{C}} \otimes \ketbra{000}{000}_{T_{A}T_{B}T_{C}} \\
    + \frac{1}{2}\ketbra{\Phi_{0}}{\Phi_{0}}_{Q_{A}Q_{C}} \otimes \ketbra{+}{+}_{Q_{B}} \otimes \ketbra{111}{111}_{T_{A}T_{B}T_{C}}, \label{eq:stateHon}
\end{multline}
where $Q_{A/B/C}$ are qubit systems, $\ket{\Phi_{0}} = (\ket{00}+\ket{11})/\sqrt{2}$, $\ket{\pm} = (\ket{0}\pm \ket{1})/\sqrt{2}$ and $T_{A}$, $T_{B}$ and $T_{C}$ are binary registers. Each party then performs the following ideal, four outcome measurements. On registers $T_A$, $T_B$ and $T_C$, each party measures in the computational basis (for every input). On registers $Q_A$, $Q_B$ and $Q_C$, the parties perform qubit measurements according to their input, described by the observables 
\begin{equation}
    \begin{aligned}
        A_{0} &= \sigma_{Z},\ A_{1} = \sigma_{X}, \\
        B_{0/1} &= C_{0/1} = (\sigma_{Z} \pm \sigma_{X})/\sqrt{2}, \
        B_{2} = C_{2} = \sigma_{Z},
    \end{aligned}
\end{equation}
where $\sigma_{Z}$ and $\sigma_{X}$ are the Pauli operators. Note that the classical nature of the $T_{A/B/C}$ registers makes it easy not to abort in step 4 ($\ketbra{0}{0}_{T_{A}}$ etc.\ can be macroscopic).

\ifarxiv\subsection{Security proof} \label{sec:sec}\else\medskip\noindent{\it Security proof.---}\fi
In the asymptotic regime, there are typically two quantities of interest for a DICKA protocol. The first is the conditional von Neumann entropy of Alice's key generation measurement, stored in the classical system $A$, conditioned on Eve's quantum side information $E$. This captures the amount of extractable randomness per round in Alice's raw string. The second captures the worst case cost of error reconciliation between Alice and every other party. That is, it captures the maximum information Alice needs to release so that every party can correct their raw key to Alice's. The conference key rate is the difference between these two terms, and captures the amount of distillable key, per round, in the asymptotic limit. If one entangled state is consumed in each round of the protocol, as in the case of GHZ based protocols~\cite{Ribeiro18,Holz19,Ribeiro19}, this can be reinterpreted as the rate per entangled state. The suitability of this definition is debatable, which we elaborate on in \ifarxiv\cref{sec:alt_a,sec:disc}\else the Discussion and Supplemental Material~\cite{supp}\fi.

For our protocol, we define the bipartite rate between Alice and Bob, conditioned on a particular value of the flag $T$\ifarxiv\else~\fi\footnote{Recall, given that the protocol did not abort, all parties agree on their classical variable $T_{A/B/C} \equiv T$.}, by  
\begin{multline}
    r^{\infty}_{T=0} = \inf H(A|X=0,T=0,E) \\ - H(A|B,X=0,Y=2,T=0),
\end{multline}
where $H(A|E)_{\rho} = H(\rho_{AE}) - H(\rho_{E})$, $H(\rho) = -\mathrm{Tr}[\rho \log(\rho)]$ and the infimum is taken over all quantum states and measurements compatible with some observations, such as the violation of a Bell inequality. A similar quantity, $r^{\infty}_{T=1}$, can be defined between Alice and Carole. Following Alice's XOR reconciliation procedure, the final conference key rate is given by
\begin{equation}
    r^{\infty}_{\mathrm{CKA}} = \min \{ p_{T}(0) \,  r^{\infty}_{T=0},p_{T}(1)\, r^{\infty}_{T=1}\},
\end{equation}
where $p_{T}(t)$ is the probability the parties record $T=t$. Since in each round a single multi-partite state is distributed, $r^{\infty}_{\mathrm{CKA}}$ captures the asymptotic rate per state consumed by the protocol. Proving security then becomes the task of bounding $\inf H(A|X=0,T=t,E)$ for $t \in \{0,1\}$.

For the security proof, we assume each party holds an arbitrary quantum system $\tilde{Q}_{A/B/C}$. We denote Alice's POVM elements $\tilde{M}_{(a,t)|x}$, Bob's $\tilde{N}_{(b,t)|y}$ and Carole's $\tilde{O}_{(c,t)|z}$. We denote the underlying state, including Eve's purification, by $\ket{\Psi}_{\tilde{Q}_{A}\tilde{Q}_{B}\tilde{Q}_{C}E}$. The outcome $t$ can have no dependence on the parties' input choices, otherwise the no-signalling assumption would be violated. Therefore, we can consider the parties observing a joint distribution (suppressing the tensor product)
\begin{equation}
    p(abct|xyz) = \bra{\Psi}\tilde{M}_{(a,t)|x} \tilde{N}_{(b,t)|y} \tilde{O}_{(c,t)|z} \ket{\Psi}, \label{eq:qProb}
\end{equation}
from which the Bell value  $\langle I \rangle := \langle I_{\mathrm{CHSH}}^{AB,T=0} \rangle + \langle I_{\mathrm{CHSH}}^{AC,T=1} \rangle$ can be calculated, where $I_{\mathrm{CHSH}}^{AB,T=t} := \tilde{A}_{0,t}(\tilde{B}_{0,t} + \tilde{B}_{1,t}) + \tilde{A}_{1,t}(\tilde{B}_{0,t}-\tilde{B}_{1,t})$, and similarly for $I_{\mathrm{CHSH}}^{AC,T=t}$. In the above, $\tilde{A}_{x,t} := \tilde{M}_{(0,t)|x} - \tilde{M}_{(1,t)|x}$ (note that $\tilde{A}_{x,t}$ is not an observable) and similarly for $\tilde{B}_{y,t},\tilde{C}_{z,t}$, and for any operator $O$, $\langle O \rangle = \bra{\Psi} O \ket{\Psi}$. The local bound of $\langle I \rangle$ is at most $2$, whereas the quantum bound of $2\sqrt{2}$ is achieved by the honest strategy. 

The following proposition underpins the security of our construction at maximum violation, i.e., when $\langle I \rangle = 2\sqrt{2}$, and under the assumption of asymptotically many independent and identically distributed rounds (this can be used as a basis for finite rates using tools such as the entropy accumulation theorem~\cite{DFR,ADFRV}).

\ifarxiv\begin{proposition}[informal] \label{prop:ent}\else\medskip\emph{Proposition~1.---}\fi
    Let the Bell expression $I$ be defined above. Then we have the following for $t \in \{0,1\}$:
    \begin{equation}
        \inf \, H(A|X=0,T=t,E) = 1, 
    \end{equation} 
    where the infimum is taken over all quantum states and measurements that achieve $\langle I \rangle = 2\sqrt{2}$ and for which it is impossible to abort in Step~4, and the von Neumann entropy is evaluated on the post-measurement state
    \begin{multline}
        \rho_{AE|X=0,T=t} = \frac{1}{p_{T|X=0}(t)}\sum_{a \in \{0,1\}} \ketbra{a}{a}_{A} \\ 
        \otimes  \mathrm{Tr}_{\tilde{Q}_{A}\tilde{Q}_{B}\tilde{Q}_{C}}\Big[ (\tilde{M}_{(a,t)|0}\otimes \id_{\tilde{Q}_{B}\tilde{Q}_{C}E})\ketbra{\Psi}{\Psi}\Big], \label{eq:pms} 
    \end{multline}
    where $p_{T|X=0}(t) := \sum_{a\in \{0,1\}} \bra{\Psi}(\tilde{M}_{(a,t)|0}\otimes \id_{\tilde{Q}_{B}\tilde{Q}_{C}E})\ket{\Psi}$.
\ifarxiv\end{proposition}\else\medskip\fi

\noindent The precise proposition and proof can be found in \ifarxiv\cref{app:proof}\else the Supplemental Material~\cite{supp}\fi. Note that in the honest strategy, Alice and Bob (Carole) will observe perfect correlations when $X=0$, $Y=2$ and $T=0$ ($X=0$, $Z=2$ and $T=1$), hence we find $r^{\infty}_{T=0} = r^{\infty}_{T=1} = 1$. The honest strategy also has $p_{T}(0) = p_{T}(1) = 1/2$ and as a result we find $r^{\infty}_{\mathrm{CKA}} = 1/2$. We also show our construction is robust to noise in \ifarxiv \cref{sec:rob}. \else the Supplemental Material~\cite{supp}. \fi

Our protocol is similar in structure to the device-dependent result of Ref.~\cite{Carrara21}. Moreover, it leaves open the possibility for modification. For example, one could equivalently use an honest implementation without source randomness, or substitute the CHSH inequality for another Bell inequality with different properties. We elaborate on these points in \ifarxiv \cref{sec:alt}. \else the Supplemental Material~\cite{supp}. \fi

\ifarxiv\section{Discussion}\label{sec:disc}\else\medskip\noindent{\it Discussion.---}\fi
In this work we presented a secure DICKA protocol that can be performed (even optimally) without GME states, resolving the open question posed in Ref.~\cite{Grasselli_2023}. Our protocol is operationally equivalent to concatenating bipartite DIQKD protocols based on CHSH tests. This prompts another question: how do protocols with bipartite entanglement compare to those with GME? Frequently, the figure of merit chosen for this comparison is the key rate per network resource, or rather, per entangled state~\cite{MurtaReview20, Carrara21,Grasselli_2023}. As acknowledged by the authors of Ref.~\cite{Grasselli_2023}, this can be misleading, since it implies that every entangled state has the same experimental cost. In many experimental setups, it is easier to generate bipartite entanglement than entanglement shared by three or more systems, and hence it would not make sense to define the same cost to both (as done in Ref.~\cite{Grasselli_2023}). Theoretically, a more suitable measure could consist of an entanglement monotone~\cite{Plenio07}, such as the amount of distillable entanglement~\cite{Arnon_Friedman_2019}, which is likely complicated in the multi-partite scenario~\cite{Philip25}. In practice, the best approach will be architecture dependent, with factors such as cost, complexity and performance playing a role.      

For example, Ref.~\cite{Epping_2017} discussed different network topologies in which (device-dependent) protocols using GHZ states are argued to have an advantage over the bipartite case below certain noise thresholds. Here, the key rate is measured per unit time. This depends on the network structure, something that is missed when counting per entangled state~\cite{Grasselli_2023}. In Ref.~\cite{Pickston_2023}, a different network topology was used that allows multiple Bell pairs to be distilled in a single network use, referred to as multi-cast. Under these conditions, the experiment in Ref.~\cite{Pickston_2023} maintains a key rate improvement from the use of GHZ states. Whether such improvements can be observed in DICKA is an interesting question, since the stringent noise requirements of DI protocols may be limiting.           

Additional arguments against concatenating bipartite DIQKD protocols include~\cite{Epping_2017} $(i)$ the additional qubit consumption of $N-1$ Einstein-Podolsky-Rosen pairs versus one GHZ state, $(ii)$ the additional classical communication cost to perform efficient XOR based reconciliation, and $(iii)$, the issue of memory effects~\cite{Grasselli_2023}, requiring Alice to use a separate device to establish a key with each party. Points $(i)$ and $(ii)$ depend on the capabilities of an experimental set up. For example, it may take multiple attempts to successfully distribute a GHZ state, and a greater classical overhead may be acceptable in return for more noise-resilient quantum resources. Moreover, a significant bottleneck for QKD implementations is the computational efficiency of privacy amplification~\cite{Hayashi_2016}. Efficient algorithms impose requirements on the min-entropy rate (the min-entropy accumulated as a fraction of the total block length) and the seed length. Whether bipartite protocols or multi-partite protocols are more beneficial in this regard is a further open question.

For point $(iii)$, we have introduced a DICKA protocol based on bipartite entanglement that inherently requires one device per party. In essence, the conditional entropy of Alice's key when entanglement is shared with Bob is certified by the same Bell test as when shared with Carole, so Alice can use the same device. If the device was to leak information about an Alice-Bob key bit in the announced data for the Alice-Carole key, this would decrease the Bell violation, as Alice and Carole would no longer maximally violate CHSH. We also acknowledge that concatenating bipartite DIQKD in the standard way with a single device per party can still be made secure against memory attacks. Though penalizing as the number of parties increase, provided \emph{all} information leaked over the classical channel is accounted for in the min-entropy estimate of each bipartite key, the final conference key remains secure.      

Finally, we emphasize that this discussion is relevant to DI protocols that use GHZ states. This excludes, for example, protocols based on interfering single photons or weak coherent pulses~\cite{Grasselli_2019,Cao21,Cao_2021b, Carrara23}. Similar remarks could also be made for DI randomness generation protocols based on GHZ states~\cite{WBC2,GrasselliMulti21,Woodhead_2018,Grasselli_2023}, versus, e.g., parallel bipartite protocols. It would be interesting to see if the protocol presented here could be adapted to this task, allowing randomness generation from $N$ devices using bipartite entanglement. Unlike key distribution however, randomness expansion experiments do not require the physical separation of devices. Rather, the no-signalling assumption can be satisfied by shielding the devices from each other, whilst remaining in the same lab. This reduces the effect of transmission loss when distributing entanglement, boosting the noise tolerance. This could make protocols that rely on GHZ states more realistic. 

\ifarxiv\acknowledgements\else\medskip\noindent{\it Acknowledgements.---}\fi The authors thank Gláucia Murta for useful discussions about CKA protocols and feedback on a preliminary version of the manuscript. This work was supported by EPSRC via the Quantum Communications Hub (Grant No.\ EP/T001011/1), the Integrated Quantum Networks Hub (Grant No.\ EP/Z533208/1), Grant No.\ EP/SO23607/1, and the European Union's Horizon Europe research and innovation programme under the project ``Quantum Secure Networks Partnership'' (QSNP, grant agreement No.\ 101114043).

%

\onecolumngrid
\ifarxiv\appendix\else\renewcommand{\thesection}{S\arabic{section}}\setcounter{equation}{0}\renewcommand{\theequation}{S\arabic{equation}}   \newpage\begin{center}\large\textbf{Supplemental Material for Bipartite entanglement is sufficient for standard device-independent conference key agreement}\end{center}\fi

\section{Discussion of information flow in the protocol}
\label{app:info}
The security of our DICKA protocol is derived from the violation of a Bell inequality. Bell inequalities were originally introduced as a way to falsify the description of quantum theory in terms of local hidden variables (LHV), and such experiments are known as Bell tests~\cite{EPR,Bell_book,CHSH}. Conclusively showing that no LHV model can explain the correlations observed in a Bell test requires that certain loopholes are closed in the experimental implementation. A loophole is any means by which an LHV model could fake a Bell violation. A key example is the locality loophole~\cite{Aspect82,Hensen2015}. If one party's measurement choice can influence the measurement outcome of the other, violating a Bell inequality becomes trivial. This loophole is typically closed by space-like separating the measurements of each party, and relying on the impossibility of superluminal signalling as implied by relativity theory (if the measurement settings of one party were to influence the measurement outcome of the other, the communication would have to take place faster than the speed of light). Thus, the no-signalling condition is a consequence of the underlying causal structure. When space-like separation is not satisfied, the locality loophole is said to be open. 

However, in the scenario we consider, space-like separation is not enough to guarantee security, and a stronger condition is required. For any cryptographic task to be secure, it is necessary that secret information generated inside the laboratory of the users during the protocol (such as raw key) is not leaked (even at below the speed of light). This cannot be guaranteed by space-like separation. Instead, the assumption that the lab is secure is used to prevent private information leaking to the outside world. In the DICKA scenario, unwanted communication of each party's measurement setting is covered by this assumption, and space-like separation is no longer necessary for security.

Another important consideration for our protocol, and any other DIQKD/DICKA protocol based on spot checking, is the communication of whether each round is ``test'' or ``generation'' by Alice to every other party (step 1). It is crucial that the state creation device and the measurement devices of each party do not know this information. Otherwise, the devices could behave honestly on test rounds and follow a pre-programmed classical behaviour on generation rounds. Such a strategy could ensure that the protocol (almost always) does not abort, yet the raw key is completely insecure. To prevent this, each party should store their share of the distributed state locally before Alice randomly assigns the round as ``test'' or ``generation'' and communicates this to the other parties. Each party should then sample their input according to the appropriate distribution and measure \emph{without} allowing the information about whether the round is a ``test'' or ``generation'' round to propagate to their measurement device.

For many experimental platforms, local storage is challenging with current technology (although this was achieved in a recent experimental demonstration of DIQKD~\cite{Nadlinger_2022}). An alternative is to modify the protocol to one in which each party uses a biased local random number generator, where a larger weight is placed on choosing measurement settings associated with generation rounds. This was discussed in~\cite{bhavsar2023improved} in the context of DI quantum randomness expansion.

A second alternative would be to use a secure channel to communicate ``test'' or ``generate'' prior to the state being shared. This could be done using pre-shared conference key and hence the protocol would become conference key expansion and the amount of pre-shared conference key would need to be accounted for in the rate. In the asymptotic case this is negligible because the probability of a test round can be made arbitrarily small. The channel could also be secured using a key generated with a classical protocol believed to be secure against a quantum attack. The final conference key would then have everlasting security, provided the classical protocol was not broken in the time between sending the ``test''/``generate'' information and measuring the quantum state, which can be kept very short.

\section{Proof of \ifarxiv\cref{prop:ent}\else Proposition~1\fi} \label{app:proof}

\subsection{Notation and definitions}
We denote the physical Hilbert space of all four parties $\mathcal{H} = \mathcal{H}_{\tilde{Q}_{A}}\otimes \mathcal{H}_{\tilde{Q}_{B}} \otimes \mathcal{H}_{\tilde{Q}_{C}} \otimes \mathcal{H}_{E}$. For any operators $M,N,O$ on $\mathcal{H}_{\tilde{Q}_{A}},\mathcal{H}_{\tilde{Q}_{B}}, \mathcal{H}_{\tilde{Q}_{C}}$ respectively, we implicitly include a tensor product with identity on all other subsystems. That is, for a state $\ket{\psi} \in \mathcal{H}$, $MNO\ket{\psi}$ is understood as $M \otimes N \otimes O \otimes \id_{E}\ket{\psi}$. We remove the tilde and write $Q_{A},Q_{B},Q_{C}$ for some target quantum system, e.g., a qubit. 

A tuple $(\mathcal{H}_{\tilde{Q}_{A}} , \mathcal{H}_{\tilde{Q}_{B}} ,\mathcal{H}_{E}, \ket{\Psi},\{\{\tilde{P}_{a|x}\}_{a}\}_{x},\{\{\tilde{Q}_{b|y}\}_{b}\}_{y})$ is a tensor product quantum model if $\ket{\Psi} \in \mathcal{H}_{\tilde{Q}_{A}} \otimes \mathcal{H}_{\tilde{Q}_{B}} \otimes \mathcal{H}_{E}$, $\tilde{P}_{a|x} \in \mathrm{P}(\mathcal{H}_{\tilde{Q}_{A}})$ and $\tilde{Q}_{b|y} \in \mathrm{P}(\mathcal{H}_{\tilde{Q}_{B}})$, where $\mathrm{P}(\mathcal{H})$ is the set of positive linear operators on a Hilbert space $\mathcal{H}$, and $\sum_{a}\tilde{P}_{a|x} = \id_{\tilde{Q}_{A}}, \ \sum_{b}\tilde{Q}_{b|y} = \id_{\tilde{Q}_{B}}$. The tuple $(\mathcal{H}_{\tilde{Q}_{A}} , \mathcal{H}_{\tilde{Q}_{B}} ,\mathcal{H}_{E}, \ket{\Psi},\{\{\tilde{P}_{a|x}\}_{a}\}_{x},\{\{\tilde{Q}_{b|y}\}_{b}\}_{y})$ is a quantum model for a behaviour $p_{AB|XY}(ab|xy)$ if it additionally satisfies $\langle\tilde{P}_{a|x}\tilde{Q}_{b|y}\rangle = p_{AB|XY}(ab|xy)$, where $\langle\tilde{P}_{a|x}\tilde{Q}_{b|y}\rangle$ is shorthand for $\bra{\Psi}\tilde{P}_{a|x}\tilde{Q}_{b|y}\ket{\Psi}$. When each party has binary inputs and binary outputs, we denote the set of such (bipartite) behaviours $\mathcal{Q}$. The definition generalizes straightforwardly to more parties. Another class of quantum models are those consisting of a single Hilbert space for the honest parties $\mathcal{H}_{\tilde{Q}}$, along with commuting operators $\tilde{P}_{a|x},\tilde{Q}_{b|y} \in \mathrm{P}(\mathcal{H}_{\tilde{Q}})$, rather than assuming a tensor product structure. Note we still assume a tensor product structure with Eve, i.e., $\ket{\Psi} \in \mathcal{H}_{\tilde{Q}} \otimes \mathcal{H}_{E}$. We call such models commuting operator models, and they include tensor product models as a special case.    

\subsection{Proof}
As discussed in the main text, we prove \ifarxiv\cref{prop:ent}\else Proposition~1 \fi under the assumption of asymptotically many independent and identically distributed rounds. Without loss of generality, we also assume that all measurements are projective according to Naimark’s dilation theorem~\cite{paulsen_2003}. Before giving the proof, we define the projector onto the outcome $t$ for each party, conditioned on their input:
    \begin{equation}
        \tilde{\Pi}_{t|x}^{A} := \tilde{M}_{(0,t)|x} + \tilde{M}_{(1,t)|x}. \label{eq:proj}
    \end{equation}
We can similarly define $\tilde{\Pi}_{t|y}^{B},\tilde{\Pi}_{t|z}^{C}$, and since $\sum_{a,t}\tilde{M}_{(a,t)|x} = \id_{\tilde{Q}_{A}}$, $\sum_{b,t}\tilde{N}_{(b,t)|y} = \id_{\tilde{Q}_{B}}$ etc., we can see
    \begin{equation}
        \sum_{t}\tilde{\Pi}_{t|x}^{A} = \id_{\tilde{Q}_{A}}, \ \ \   \sum_{t}\tilde{\Pi}_{t|y}^{B} = \id_{\tilde{Q}_{B}}, \ \text{and} \  \sum_{t}\tilde{\Pi}_{t|z}^{C} = \id_{\tilde{Q}_{C}}.   \label{eq:proj1}
    \end{equation} 
    We then define the set of quantum models $(\mathcal{H}_{\tilde{Q}_{A}} , \mathcal{H}_{\tilde{Q}_{B}} ,\mathcal{H}_{\tilde{Q}_{C}}, \mathcal{H}_{E}, \ket{\Psi},\{\{\tilde{M}_{(a,t)|x}\}_{a,t}\}_{x},\{\{\tilde{N}_{(b,t)|y}\}_{b,t}\}_{y},\{\{\tilde{O}_{(c,t)|z}\}_{c,t}\}_{z})$ satisfying
    \begin{equation}
        0<\langle \tilde{\Pi}^{A}_{t|x} \rangle = \langle \tilde{\Pi}^{B}_{t|y} \rangle = \langle \tilde{\Pi}^{C}_{t|z} \rangle = \langle \tilde{\Pi}^{A}_{t|x}\tilde{\Pi}^{B}_{t|y} \rangle = \langle \tilde{\Pi}^{A}_{t|x}\tilde{\Pi}^{C}_{t|z} \rangle = \langle \tilde{\Pi}^{B}_{t|y}\tilde{\Pi}^{C}_{t|z} \rangle = \langle \tilde{\Pi}^{A}_{t|x}\tilde{\Pi}^{B}_{t|y}\tilde{\Pi}^{C}_{t|z} \rangle =: p_{T}(t) \label{eq:projConstr1}
    \end{equation}
    and $\langle I \rangle = 2\sqrt{2}$, where $I = I_{\mathrm{CHSH}}^{AB,T=0} + I_{\mathrm{CHSH}}^{AC,T=1}$ is the Bell expression defined in the main text, as $\mathcal{F}$. Note that all quantum models that satisfy \cref{eq:projConstr1} cannot cause the protocol to abort on step 4 in the main text. That is, all parties must obtain the same value of the flag each round (i.e., $T_{A} = T_{B} = T_{C}$), and $p_{T}(0),p_{T}(1) > 0$. This condition is achievable in practice, since in the honest implementation, the flag corresponds a classical signal which can be sent robustly down the channel.  
    
    Below, we give a lemma that states some useful properties of quantum models in $\mathcal{F}$.
    \begin{lemma}
        Let $(\mathcal{H}_{\tilde{Q}_{A}} , \mathcal{H}_{\tilde{Q}_{B}},\mathcal{H}_{\tilde{Q}_{C}} ,\mathcal{H}_{E},\ket{\Psi},\{\{\tilde{M}_{(a,t)|x}\}_{a,t}\}_{x},\{\{\tilde{N}_{(b,t)|y}\}_{b,t}\}_{y},\{\{\tilde{O}_{(c,t)|z}\}_{c,t}\}_{z})$ be any quantum model in $\mathcal{F}$ and $\tilde{\Pi}_{t|x}^{A},\tilde{\Pi}_{t|y}^{B},\tilde{\Pi}_{t|z}^{C}$ be as defined in \cref{eq:proj}. Then the following property must be satisfied for all $x,y,z,t\in \{0,1\}$$\mathrm{:}$
        \begin{equation}
            \begin{aligned}
                \tilde{\Pi}^{A}_{t|x}  \ket{\Psi} = \tilde{\Pi}^{B}_{t|y}  \ket{\Psi} = \tilde{\Pi}^{C}_{t|z}  \ket{\Psi}.
            \end{aligned} \label{eq:projprop}
        \end{equation}
       \label{lem:proj}
    \end{lemma}
    \begin{proof}
    Note the following consequence of \cref{eq:projConstr1}: 
    \begin{equation}
        \begin{aligned}
            0 &= \bra{\Psi}\tilde{\Pi}^{B}_{t|y}\ket{\Psi} - \bra{\Psi}\tilde{\Pi}^{A}_{t|x}\tilde{\Pi}^{B}_{t|y}\ket{\Psi}\\
            &= \bra{\Psi}\tilde{\Pi}^{B}_{t|y}(\id - \tilde{\Pi}^{A}_{t|x})\ket{\Psi} \\
            &= \bra{\Psi} \tilde{\Pi}^{B}_{t|y} \tilde{\Pi}^{A}_{t\oplus 1|x} \ket{\Psi},
        \end{aligned}
    \end{equation}
    where in the third equality we used the fact that $\sum_{a,t}\tilde{M}_{(a,t)|x}\ket{\Psi} = \ket{\Psi}$. We therefore see $0 = \tilde{\Pi}^{B}_{t|y} \tilde{\Pi}^{A}_{t\oplus 1|x} \ket{\Psi} = \tilde{\Pi}^{B}_{t|y} (\id-\tilde{\Pi}^{A}_{t|x}) \ket{\Psi}$, hence $\tilde{\Pi}^{B}_{t|y}  \ket{\Psi} = \tilde{\Pi}^{B}_{t|y} \tilde{\Pi}^{A}_{t|x} \ket{\Psi}$. Using a similar argument $\tilde{\Pi}^{A}_{t|x}  \ket{\Psi} = \tilde{\Pi}^{B}_{t|y} \tilde{\Pi}^{A}_{t|x} \ket{\Psi}$, which implies $\tilde{\Pi}^{B}_{t|y}  \ket{\Psi} = \tilde{\Pi}^{A}_{t|x}  \ket{\Psi}$. We also find $\tilde{\Pi}^{C}_{t|z}  \ket{\Psi} = \tilde{\Pi}^{A}_{t|x}  \ket{\Psi}$, proving the claim.
    \end{proof}

    Having established these facts, we can now present the proof of \ifarxiv\cref{prop:ent}\else Proposition~1\fi.

    \vspace{0.3cm}

\noindent \textbf{Proposition 1.} \textit{Let $\mathcal{F}$ be the class of quantum models defined above. Then we have the following:
    \begin{equation}
\inf_{\mathcal{F}} H(A|X=0,T=t,E) = 1, \label{eq:p1rate}
    \end{equation} 
    where the von Neumann entropy is evaluated on the post-measurement state
    \begin{equation}
        \rho_{AE|X=0,T=t} = \frac{1}{p_{T|X=0}(t)}\sum_{a \in \{0,1\}} \ketbra{a}{a}_{A} \\ 
        \otimes  \mathrm{Tr}_{\tilde{Q}_{A}\tilde{Q}_{B}\tilde{Q}_{C}}\Big[ (\tilde{M}_{(a,t)|0}\otimes \id_{\tilde{Q}_{B}\tilde{Q}_{C}E})\ketbra{\Psi}{\Psi}\Big], 
    \end{equation}\label{eq:app_pms}
    where $p_{T|X=0}(t) := \sum_{a \in \{0,1\}}\bra{\Psi} (\tilde{M}_{(a,t)|0}\otimes \id_{\tilde{Q}_{B}\tilde{Q}_{C}E})\ket{\Psi}$.}
\begin{proof}    
    We begin by finding an SOS decomposition for $2\sqrt{2}\id - I$. To achieve this, we find decompositions for $I_{\mathrm{CHSH}}^{AB,t}$ and $I_{\mathrm{CHSH}}^{AC,t}$ by modifying an existing decomposition for the CHSH functional (see, e.g.,~\cite{BampsPironio,Cui_2020}). By direct calculation,
    \begin{equation}
        0 \preceq \frac{\sqrt{2}}{4}(\tilde{A}_{0,t} + \tilde{A}_{1,t} -\sqrt{2}\tilde{B}_{0,t})^{2} 
        + \frac{\sqrt{2}}{4}(\tilde{A}_{0,t} - \tilde{A}_{1,t} -\sqrt{2}\tilde{B}_{1,t})^{2}  = \frac{\sqrt{2}}{2}(\tilde{\Pi}_{t|0}^{A}+\tilde{\Pi}_{t|1}^{A}+\tilde{\Pi}_{t|0}^{B}+\tilde{\Pi}_{t|1}^{B}) - I_{\mathrm{CHSH}}^{AB,t}. \label{eq:sos1}
    \end{equation}
    The equality holds from the definition of the operators $\tilde{A}_{x,t},\tilde{B}_{y,t}$, e.g., $(\tilde{A}_{x,t})^{2} = \tilde{M}_{(0,t)|x} + \tilde{M}_{(1,t)|x} = \tilde{\Pi}_{t|x}^{A}$ (recall we have extended the POVMs to be orthogonal projectors). More intuitively, the above equation implies
    \begin{equation}
        \frac{\sqrt{2}}{2} \langle \tilde{\Pi}_{t|0}^{A}+\tilde{\Pi}_{t|1}^{A}+\tilde{\Pi}_{t|0}^{B}+\tilde{\Pi}_{t|1}^{B}\rangle  =  2\sqrt{2} \, p_{T}(t) \geq \langle  I_{\mathrm{CHSH}}^{AB,t}\rangle, \label{eq:tbound2}
    \end{equation}
    where we used relation $(i)$ of \cref{lem:proj}. In words, the Tsirelson bound of $I_{\mathrm{CHSH}}^{AB,t}$, for quantum models that do not cause the protocol to abort, is given by the maximum CHSH violation, weighted by the probability of observing the outcome $T=t$. A similar decomposition can be written for $I^{AC,t}_{\mathrm{CHSH}}$. Let us write $\tilde{\Gamma}^{AB}_{t} = \tilde{\Pi}_{t|0}^{A}+\tilde{\Pi}_{t|1}^{A}+\tilde{\Pi}_{t|0}^{B}+\tilde{\Pi}_{t|1}^{B}$, and similarly $\tilde{\Gamma}^{AC}_{t}$. Then, by the above,
    \begin{equation}
        \begin{aligned}
            \frac{\sqrt{2}}{2}(\tilde{\Gamma}^{AB}_{0} + \tilde{\Gamma}^{AC}_{1}) - I = \Big(\frac{\sqrt{2}}{2}\tilde{\Gamma}^{AB}_{0} - I_{\mathrm{CHSH}}^{AB,0}\Big) + \Big(\frac{\sqrt{2}}{2}\tilde{\Gamma}^{AC}_{1} - I_{\mathrm{CHSH}}^{AC,1}\Big)  = \sum_{i}P_{i}^{\dagger}P_{i} + \sum_{j}Q_{j}^{\dagger}Q_{j}\succeq 0,
        \end{aligned} \label{eq:SOS}
    \end{equation}
    where $\{P_{i}\}_{i},\{Q_{j}\}_{j}$ collect all the polynomials used to define the SOS decompositions for $\frac{\sqrt{2}}{2}\tilde{\Gamma}^{AB}_{0} - I^{AB,t=0}_{\mathrm{CHSH}}$ and $\frac{\sqrt{2}}{2}\tilde{\Gamma}^{AC}_{1} -I^{AC,t=1}_{\mathrm{CHSH}}$, respectively, i.e., from \eqref{eq:sos1} we have $P_{0} = 2^{-3/4}(\tilde{A}_{0,t} + \tilde{A}_{1,t} -\sqrt{2}\tilde{B}_{0,t})$, $P_{1} = 2^{-3/4}(\tilde{A}_{0,t} - \tilde{A}_{1,t} -\sqrt{2}\tilde{B}_{1,t})$ etc. Notice 
    \begin{equation}
        (\tilde{\Gamma}_{0}^{AB} + \tilde{\Gamma}_{1}^{AC})\ket{\Psi} = 4\ket{\Psi},
    \end{equation}
    where we applied \cref{lem:proj} and \eqref{eq:proj1}. Therefore, given the protocol did not abort, 
    \begin{equation}
        \bar{I}\ket{\Psi} := (2\sqrt{2}\id - I)\ket{\Psi} = \sum_{i}P_{i}^{\dagger}P_{i}\ket{\Psi} + \sum_{j}Q_{j}^{\dagger}Q_{j}\ket{\Psi}.
    \end{equation}
    In particular, we see that $I$ has a Tsirelson bound of $2\sqrt{2}$ in this case. 
    
    Now, suppose the protocol did not abort and we achieve this bound, i.e., observe $\langle I \rangle = 2\sqrt{2}$. Then we can immediately conclude
    \begin{equation}
        2\sqrt{2} - \langle I \rangle =  \Big\langle\sum_{i}P_{i}^{\dagger}P_{i} \Big\rangle + \Big\langle\sum_{j}Q_{j}^{\dagger}Q_{j}\Big\rangle = 0.
    \end{equation}
    As a result, the expectation of each SOS polynomial must be equal to zero. Using \cref{eq:tbound2}, we find (recall the protocol aborts if $p_{T}(0)  = 0$ or $p_{T}(1) = 0$)
    \begin{equation}
    \begin{aligned}
        \frac{\langle I_{\mathrm{CHSH}}^{AB,0} \rangle}{p_{T}(0)} = 2\sqrt{2}, \ \ \mathrm{and} \ \ \frac{\langle I_{\mathrm{CHSH}}^{AC,1} \rangle}{p_{T}(1)} = 2\sqrt{2}.
    \end{aligned}
     \label{eq:CHSHconst}
    \end{equation}
    
    Let us define the conditional distributions:
    \begin{equation}
    \begin{aligned}
        p_{AB|XY,T=t}(ab|xy) &:= \frac{1}{p_{T}(t)}\sum_{c}p(abct|xyz) = \frac{1}{p_{T}(t)}\bra{\Psi} \tilde{\Pi}_{t|z}^{C}\tilde{M}_{(a,t)|x}\tilde{N}_{(b,t)|y}\tilde{\Pi}_{t|z}^{C}\ket{\Psi}, \\
        p_{AC|XZ,T=t}(ac|xz) &:= \frac{1}{p_{T}(t)}\sum_{b}p(abct|xyz) = \frac{1}{p_{T}(t)}\bra{\Psi} \tilde{\Pi}_{t|y}^{B}\tilde{M}_{(a,t)|x}\tilde{O}_{(c,t)|z}\tilde{\Pi}_{t|y}^{B}\ket{\Psi},
    \end{aligned}
     \label{eq:prob}
    \end{equation}
    where we used the fact that $\sum_{c}\tilde{O}_{(c,t)|z}\ket{\Psi} = \tilde{\Pi}_{t|z}^{C}\ket{\Psi}$, and that $\tilde{\Pi}_{t|z}^{C}$ is a projector which commutes with $\tilde{M}_{(a,t)|x}\tilde{N}_{(b,t)|y}$. Similar steps were taken for $\tilde{N}_{(b,t)|y}$ in the second line. Note the choice of $z$ is irrelevant in the first definition (likewise $y$ in the second) following the no-signalling conditions.

    Inserting these into the definitions of $\langle I_{\mathrm{CHSH}}^{AB,0} \rangle$ and $\langle I_{\mathrm{CHSH}}^{AC,1} \rangle$ we have
    \begin{equation}
        \frac{\langle I_{\mathrm{CHSH}}^{AB,0} \rangle}{p_{T}(0)} = \sum_{abxy} (-1)^{a+b+xy}p_{AB|T=0}(ab|xy), \ \ \mathrm{and} \ \ \frac{\langle I_{\mathrm{CHSH}}^{AC,1} \rangle}{p_{T}(1)} = \sum_{acxz} (-1)^{a+c+xz}p_{AC|T=1}(ac|xz).
    \end{equation}
    Combined with \cref{eq:CHSHconst}, the distributions $\{p_{AB|XY,T=0}(ab|xy)\}_{abxy},\{p_{AC|XZ,T=1}(ac|xz)\}_{acxz}$ maximally violate the CHSH inequality, which is well known for its self-testing properties (see e.g.,~\cite{BampsPironio}). In particular, quantum behaviours that achieve the maximum are extreme (and nonlocal) in the quantum set $\mathcal{Q}$. In the remainder of the proof, we will show that both $p_{AB|XY,T=0}$ and $p_{AC|XZ,T=1}$ are quantum realizable, and hence extremal, which can be used to prove security using the results of Ref.~\cite{Franz_2011}.

    Towards this end, let
    \begin{equation}
        \ket{\Psi_{0}} := \frac{\tilde{\Pi}_{0|z}^{C} \ket{\Psi}}{\sqrt{p_{T}(0)}}, \ \ \mathrm{and} \ \ \ket{\Psi_{1}} := \frac{\tilde{\Pi}_{1|y}^{B} \ket{\Psi}}{\sqrt{p_{T}(1)}}, \label{eq:projState}
    \end{equation}
    allowing us to write
    \begin{equation}
    \begin{aligned}
    p_{AB|XY,T=0}(ab|xy) = \bra{\Psi_{0}}\tilde{M}_{(a,0)|x}\tilde{N}_{(b,0)|y}\ket{\Psi_{0}}, \ \ \mathrm{and} \ \ 
    p_{AC|XZ,T=1}(ac|xz) = \bra{\Psi_{1}}\tilde{M}_{(a,1)|x}\tilde{O}_{(c,1)|z}\ket{\Psi_{1}}. \label{eq:probs}
    \end{aligned}
    \end{equation}  
    For every quantum model $(\mathcal{H}_{\tilde{Q}_{A}},\mathcal{H}_{\tilde{Q}_{B}},\mathcal{H}_{\tilde{Q}_{C}},\mathcal{H}_{E},\ket{\Psi},\{\{\tilde{M}_{(a,t)|x}\}_{a,t}\}_{x},\{\{\tilde{N}_{(b,t)|y}\}_{b,t}\}_{y},\{\{\tilde{O}_{(c,t)|z}\}_{c,t}\}_{z})$, define the subset of $\mathcal{H} = \mathcal{H}_{\tilde{Q}_{A}} \otimes \mathcal{H}_{\tilde{Q}_{B}} \otimes \mathcal{H}_{\tilde{Q}_{C}} \otimes \mathcal{H}_{E}$:
    \begin{equation}
    \begin{aligned}
        \mathcal{H}_{0} &:= \Big \{\tilde{\Pi}_{0|z}^{C}\ket{\phi} \ : \ \ket{\phi} \in \mathcal{H}, \ \tilde{\Pi}^{A}_{t|x}  \ket{\phi} = \tilde{\Pi}^{B}_{t|y}  \ket{\phi} = \tilde{\Pi}^{C}_{t|z} \ket{\phi} \ \forall x,y,z,t \Big\} \subset \mathcal{H}, \\
        \mathcal{H}_{1} &:= \Big \{\tilde{\Pi}_{1|y}^{B}\ket{\phi} \ : \ \ket{\phi} \in \mathcal{H}, \ \tilde{\Pi}^{A}_{t|x}  \ket{\phi} = \tilde{\Pi}^{B}_{t|y}  \ket{\phi} = \tilde{\Pi}^{C}_{t|z} \ket{\phi} \ \forall x,y,z,t \Big\} \subset \mathcal{H}.
    \end{aligned}
    \label{eq:subspace}
    \end{equation}
    We remark that every $\ket{\psi} \in \mathcal{H}_{t}$ must then satisfy \cref{eq:projConstr1} by definition. 
    \begin{lemma}
        Let $(\mathcal{H}_{\tilde{Q}_{A}},\mathcal{H}_{\tilde{Q}_{B}}, \mathcal{H}_{\tilde{Q}_{C}},\mathcal{H}_{E},\ket{\Psi},\{\{\tilde{M}_{(a,t)|x}\}_{a,t}\}_{x},\{\{\tilde{N}_{(b,t)|y}\}_{b,t}\}_{y},\{\{\tilde{O}_{(c,t)|z}\}_{c,t}\}_{z})$ be a quantum model in $\mathcal{F}$. Let $\mathcal{H}_{t}$ and $\ket{\Psi_{t}}$ be defined with respect to that model, as given by \cref{eq:subspace} and \cref{eq:projState} respectively. Then the following statements are true: 
        \begin{enumerate}[(i)]
            \item For $t \in \{0,1\}$, the tuple $(\mathcal{H}_{t},\ket{\Psi_{t}},\{\{\tilde{M}_{(a,t)|x}\}_{a}\}_{x},\{\{\tilde{N}_{(b,t)|y}\}_{b}\}_{y},\{\{\tilde{O}_{(c,t)|z}\}_{c}\}_{z})$ is a (finite dimensional) commuting operator model. 
            \item The tuple $(\mathcal{H}_{0},\ket{\Psi_{0}},\{\{\tilde{M}_{(a,0)|x}\}_{a}\}_{x},\{\{\tilde{N}_{(b,0)|y}\}_{b}\}_{y})$ is a (finite dimensional) commuting operator model for $p_{AB|XY,T=0}(ab|xy)$. 
            \item The tuple $(\mathcal{H}_{1},\ket{\Psi_{1}},\{\{\tilde{M}_{(a,1)|x}\}_{a}\}_{x},\{\{\tilde{O}_{(c,1)|z}\}_{c}\}_{z})$ is a (finite dimensional) commuting operator model for $p_{AC|XZ,T=1}(ac|xz)$. 
        \end{enumerate} \label{lem:subspace}
    \end{lemma}
    \begin{proof}

    $\mathcal{H}_{t}$ is a subspace of $\mathcal{H}$ since it contains all linear combinations of its elements, and note that every subspace of a finite dimensional Hilbert space is closed and, in turn, is therefore also a Hilbert space~\cite{kreyszig1991introductory}. 
        
    Now, we show that every operator in the set $\{\tilde{M}_{(a,t)|x},\tilde{N}_{(b,t)|y},\tilde{O}_{(c,t)|z}\}_{a,b,c,x,y,z}$, maps $\mathcal{H}_{t}$ to itself (for a fixed $t$). To see this, consider $t=0$, and let $\ket{\psi}\in \mathcal{H}_{0}$. Then we can write $\ket{\psi} = \tilde{\Pi}^{C}_{0|z}\ket{\phi}$ for some state $\ket{\phi}$ satisfying the projection constraints in \cref{eq:subspace}. Let $\ket{\psi'} = \tilde{M}_{(a,0)|x}\ket{\psi}$ for some $a,x \in \{0,1\}$. We will show the following:
        \begin{enumerate}[(a)]
            \item $\ket{\psi'} = \tilde{\Pi}^{C}_{0|z}\ket{\phi'}$, where $\ket{\phi'} = \tilde{M}_{(a,0)|x}\ket{\phi}$.
            \item $\tilde{\Pi}^{A}_{0|x'}\ket{\phi'} = \tilde{\Pi}^{B}_{0|y}\ket{\phi'} = \tilde{\Pi}^{C}_{0|z'}\ket{\phi'}$, $\forall x',y,z'$.
        \end{enumerate}
    Point (a) can be established from the fact $\ket{\psi'} = \tilde{M}_{(a,0)|x}\ket{\psi} = \tilde{M}_{(a,0)|x}\tilde{\Pi}_{0|z}^{C}\ket{\phi} = \tilde{\Pi}_{0|z}^{C}\tilde{M}_{(a,0)|x}\ket{\phi} = \tilde{\Pi}_{0|z}^{C}\ket{\phi'}$. For (b), first note that $\ket{\phi'} = \tilde{M}_{(a,0)|x}\ket{\phi} = \tilde{M}_{(a,0)|x}\tilde{\Pi}^{A}_{0|x}\ket{\phi} = \tilde{M}_{(a,0)|x}\tilde{\Pi}^{B}_{0|y}\ket{\phi} = \tilde{\Pi}^{B}_{0|y}\ket{\phi'}$, where we used the projection constraints on $\ket{\phi}$ in \cref{eq:subspace}. Similarly, $\ket{\phi'} = \tilde{\Pi}^{C}_{0|z'}\ket{\phi'}$, and since $\tilde{M}_{(a,0)|x}\ket{\phi} = \tilde{\Pi}^{A}_{0|x}\tilde{M}_{(a,0)|x}\ket{\phi}$, we have $\ket{\phi'} = \tilde{\Pi}^{A}_{0|x}\ket{\phi'}$. What we also need to establish is for $x' \neq x$, $\ket{\phi'} = \tilde{\Pi}^{A}_{0|x'}\ket{\phi'}$. Notice, for any state $\ket{\phi}$ satisfying the projection constraints in \cref{eq:subspace},
        \begin{equation}
            \bra{\phi}\tilde{M}_{(a,0)|x}\ket{\phi} =  \bra{\phi}\tilde{\Pi}_{0|x}^{A}\tilde{M}_{(a,0)|x}\ket{\phi} = \bra{\phi}\tilde{\Pi}_{0|x'}^{A}\tilde{M}_{(a,0)|x}\ket{\phi},
        \end{equation}
        which implies 
        \begin{equation}
            \bra{\phi}(\id - \tilde{\Pi}_{0|x'}^{A})\tilde{M}_{(a,0)|x}\ket{\phi} =  0.
        \end{equation}
        The set of states $\ket{\phi}$ satisfying the projection constraints in \cref{eq:subspace} forms a Hilbert space, and since the above equation holds for all such states, we can deduce $(\id - \tilde{\Pi}_{0|x'}^{A})\tilde{M}_{(a,0)|x}\ket{\phi} =  0$, and therefore $\ket{\phi'} = \tilde{\Pi}_{0|x'}^{A}\ket{\phi'}$. This establishes (b), and therefore $\ket{\psi'} \in \mathcal{H}_{0}$, and subsequently $\tilde{M}_{(a,0)|x} \in \mathrm{P}(\mathcal{H}_{0})$. The same steps can be followed to also show $\tilde{N}_{(b,0)|y},\tilde{O}_{(c,0)|z} \in \mathrm{P}(\mathcal{H}_{0})$, and a similar argument can be made for the case $t=1$. 
        
        Finally, for every $\ket{\psi} \in \mathcal{H}_{t}$,
    \begin{equation}
        \sum_{a \in \{0,1\}}\tilde{M}_{(a,t)|x}\ket{\psi} = \tilde{\Pi}_{t|x}^{A}\ket{\psi} = \ket{\psi},
    \end{equation}
    and similarly for $\tilde{N}_{(b,t)|y}$ and $\tilde{O}_{(c,t)|z}$ (cf. \cref{eq:proj1}).
    
    We recall that the tuple
    \begin{equation}
        \mathsf{q} = \Big( \mathcal{H}_{\tilde{Q}},\mathcal{H}_{E},\ket{\Psi},\{\{\tilde{M}_{a|x}\}_{a}\}_{x},\{\{\tilde{N}_{b|y}\}_{b}\}_{y}\Big)
    \end{equation}
    is a commuting operator quantum model if $\ket{\Psi} \in \mathcal{H}_{\tilde{Q}} \otimes \mathcal{H}_{E}$, and $\{\{\tilde{M}_{a|x}\}_{a}\}_{x},\{\{\tilde{N}_{b|y}\}_{b}\}_{y}$ are sets of commuting projective measurements on $\mathcal{H}_{\tilde{Q}}$. Since all pairs of measurement operators corresponding to distinct parties commute with each other, the model in $(i)$ is a commuting operator model, proving the claim. Claims $(ii)$ and $(iii)$ follow identically.  
    \end{proof}

    It is known that the set of bipartite quantum behaviours arising from finite dimensional commuting operator models coincides with the set arising from finite dimensional tensor product models (see, e.g.,~\cite{scholz2008,Navascu_s_2012}). Thus, the behaviours $p_{AB|XY,T=0}$ and $p_{AC|XZ,T=1}$ are quantum, maximally violate the CHSH inequality, and are therefore extremal. In the following, we will make use of the following properties of extremal distributions.

    \begin{definition}[\cite{Franz_2011}]
        Let $p_{AB|XY} =\{p_{AB|XY}(ab|xy)\}$ be a behaviour in a bipartite Bell scenario. The behaviour $p_{AB|XY}$ is \emph{secure} if, for all $a,b,x,y$, $p_{AB|XY}(ab|xy) \neq p_{A|X}(a|x)p_{B|Y}(b|y)$, and the following holds for all commuting operator models compatible with $p_{AB|XY}$. For any operator $F$ on $\mathcal{H}_{E}$,
    \begin{equation}
        \bra{\Psi}\tilde{M}_{a|x} \tilde{N}_{b|y} \otimes F\ket{\Psi} = \bra{\Psi}\id \otimes F\ket{\Psi} \, p(ab|xy). \label{eq:secDef}
    \end{equation} \label{def:secure}
    \end{definition}

    \begin{lemma}[\cite{Franz_2011}]
        A behaviour $p_{AB|XY}$ is secure if and only if it is nonlocal and extremal in $\mathcal{Q}$. \label{lem:sec}
    \end{lemma}
    \noindent Based on this result, such distributions have the following decoupling property.
    \begin{lemma}
        Let $\mathsf{q} = \big( \mathcal{H}_{\tilde{Q}},\mathcal{H}_{E},\ket{\Psi},\{\{\tilde{M}_{a|x}\}_{a}\}_{x},\{\{\tilde{N}_{b|y}\}_{b}\}_{y}\big)$ be any commuting operator model that realizes a behaviour $p_{AB|XY}$ satisfying $p_{AB|XY}(ab|xy) \neq p_{A|X}(a|x)p_{B|Y}(b|y)$ for all $a,b,x,y$, and $p_{AB|XY}$ is extremal in $\mathcal{Q}$. Let $\rho_{E}^{a,b|x,y} = \mathrm{Tr}_{\tilde{Q}}[(\tilde{M}_{a|x}\tilde{N}_{b|y} \otimes \id) \ketbra{\Psi}{\Psi}]$ be the post-measurement state of Eve associated to $\mathsf{q}$ as defined above. Then for all $a,b,x,y$,
        \begin{equation}
            \rho_{E}^{a,b|x,y} = p_{AB|XY}(ab|xy)\, \rho_{E}, \label{eq:dec1_a}
        \end{equation}
        where $\rho_{E} = \mathrm{Tr}_{\tilde{Q}}[\ketbra{\Psi}{\Psi}]$. \label{lem:dec_new}
    \end{lemma}
    \begin{proof}
    We begin by applying \cref{lem:sec}. Since $p_{AB|XY}$ is nonlocal and extremal, it is secure. That is, \cref{eq:secDef} holds. By direct calculation
    \begin{equation}
        \begin{aligned}
             &\bra{\Psi}\tilde{M}_{a|x} \tilde{N}_{b|y} \otimes F\ket{\Psi} = \bra{\Psi}\id \otimes F\ket{\Psi} \, p_{AB|XY}(ab|xy)\\ 
             \iff& \mathrm{Tr}_{E}\Big[\mathrm{Tr}_{\tilde{Q}}\big[(\tilde{M}_{a|x} \tilde{N}_{b|y} \otimes F)\ketbra{\Psi}{\Psi}\big]\Big] = \mathrm{Tr}_{E}\Big[\mathrm{Tr}_{\tilde{Q}}\big[(\id_{\tilde{Q}} \otimes F)\ketbra{\Psi}{\Psi}\big]\Big] \, p_{AB|XY}(ab|xy)\\
             \iff& \mathrm{Tr}_{E}\Big[F \, \mathrm{Tr}_{\tilde{Q}}\big[(\tilde{M}_{a|x} \tilde{N}_{b|y} \otimes \id_{E})\ketbra{\Psi}{\Psi}\big]\Big] = \mathrm{Tr}_{E}\Big[F \, \mathrm{Tr}_{\tilde{Q}}\big[\ketbra{\Psi}{\Psi}\big]\Big]\, p_{AB|XY}(ab|xy)\\
             \iff& \mathrm{Tr}\Big[F\rho_{E}^{a,b|x,y}\Big] = \mathrm{Tr}\Big[F(p_{AB|XY}(ab|xy) \, \rho_{E})\Big].
             \end{aligned}
             \end{equation}
             But this holds for all $F$ if and only if $\rho_{E}^{a,b|x,y} = p_{AB|XY}(ab|xy)\, \rho_{E}$, which completes the proof.
    \end{proof}

    Consider the post-measurement state between Alice and Eve, conditioned on the input choice $X=0$:
    \begin{equation}
        \rho_{ATE|X=0} = \sum_{a,t \in \{0,1\}} \ketbra{a}{a}_{A} 
        \otimes  \ketbra{t}{t}_{T} \otimes \mathrm{Tr}_{\tilde{Q}_{A}\tilde{Q}_{B}\tilde{Q}_{C}}\Big[ \tilde{M}_{(a,t)|0}\ketbra{\Psi}{\Psi}\Big].
    \end{equation}
    Above, $A$ and $T$ denote the classical registers used to store Alice's outcomes $a$ and $t$ respectively. Conditioning on $T=t$, we consider the state 
    \begin{equation}
    \begin{aligned}
        \rho_{AE|X=0,T=t} &:= \frac{1}{p_{T}(t)}\sum_{a \in \{0,1\}} \ketbra{a}{a}_{A}  
        \otimes  \mathrm{Tr}_{\tilde{Q}_{A}\tilde{Q}_{B}\tilde{Q}_{C}}\Big[ \tilde{M}_{(a,t)|0}\ketbra{\Psi}{\Psi}\Big] \\
        &= \sum_{a \in \{0,1\}} \ketbra{a}{a}_{A}  
        \otimes  \mathrm{Tr}_{\tilde{Q}_{A}\tilde{Q}_{B}\tilde{Q}_{C}}\Big[ \tilde{M}_{(a,t)|0}\ketbra{\Psi_{t}}{\Psi_{t}}\Big], \label{eq:pms_QKD2}
    \end{aligned}
    \end{equation}
    where we used the fact that $\tilde{M}_{(a,t)|0}\ket{\Psi} = \tilde{\Pi}_{t|0}^{A}\tilde{M}_{(a,t)|0}\tilde{\Pi}_{t|0}^{A}\ket{\Psi}$ and $\ket{\Psi}$ satisfies \cref{lem:proj}. Recall that the tuple $(\mathcal{H}_{0},\mathcal{H}_{E},\ket{\Psi_{0}},\{\tilde{M}_{(a,0)|x}\}_{a},\{\tilde{N}_{(b,0)|y}\}_{b})$ realizes an extremal distribution $p_{AB|T=0}$. According to \cref{lem:dec_new}, we therefore have
    \begin{equation}
        \mathrm{Tr}_{\tilde{Q}}\Big[ (\tilde{M}_{(a,0)|x}\tilde{N}_{(b,0)|y} \otimes \id_{E})\ketbra{\Psi_{0}}{\Psi_{0}}\Big] = \hat{p}(ab|xy) \, \mathrm{Tr}_{\tilde{Q}}\Big[\ketbra{\Psi_{0}}{\Psi_{0}}\Big],
    \end{equation}
    where $\hat{p}(ab|xy)$ is the unique quantum behaviour maximally violating the CHSH inequality. Similarly for $p_{AC|T=1}$,
    \begin{equation}
        \mathrm{Tr}_{\tilde{Q}}\Big[ (\tilde{M}_{(a,1)|x}\tilde{O}_{(c,1)|z} \otimes \id_{E})\ketbra{\Psi_{1}}{\Psi_{1}}\Big] = \hat{p}(ac|xz) \, \mathrm{Tr}_{\tilde{Q}}\Big[\ketbra{\Psi_{1}}{\Psi_{1}}\Big].
    \end{equation}
    We therefore find
    \begin{equation}
        \rho_{AE|X=0,T=t} = \rho_{A|X=0,T=t} \otimes \sigma_{E|T=t}
    \end{equation}
    for some state $\sigma_{E|T=t}$, where 
    \begin{equation}
    \begin{aligned}
        \rho_{A|X=0,T=t} &= \sum_{a \in \{0,1\}} \hat{p}(a|0)\, \ketbra{a}{a}_{A} = \id_{A}/2.
    \end{aligned}
    \end{equation}
    Above, we used the fact that $\hat{p}(a|x) = \sum_{b}\hat{p}(ab|xy) = 1/2$. Putting everything together, we find
    \begin{equation}
        H(A|X=0,T=t,E)_{\rho_{AE|X=0,T=t}}= H(A|X=0,T=t)_{\rho_{A|X=0,T=t}} = 1, 
    \end{equation}
    completing the proof of \ifarxiv\cref{prop:ent}\else Proposition~1\fi.
\end{proof}

\section{Robustness} \label{sec:rob}
The robustness of our construction to noise is presented in \cref{fig:robustness}. We applied the numerical technique of Ref.~\cite{Brown2024deviceindependent}, which provides lower bounds on the optimization~\eqref{eq:p1rate} in the case of a non-maximal Bell violation, i.e., $\langle I \rangle \geq s$ for $s \in [2,2\sqrt{2}]$. For comparison, we lower bound~\eqref{eq:p1rate} by the min-entropy (this can be achieved using the NPA hierarchy~\cite{NPA1,NPA2}), and the results show the improvement offered by bounding the von Neumann entropy directly (cf.\ the solid versus dashed curves in \cref{fig:robustness}). In addition, we consider the case where, rather than just constraining the Bell value $\langle I \rangle = \langle I_{\text{CHSH}}^{AB,T=0} + I_{\text{CHSH}}^{AC,T=1} \rangle \geq s$, we constrain both terms independently, weighted according to $p_{T}(t)$. That is, we include the constraints $\langle I_{\text{CHSH}}^{AB,T=0} \rangle \geq p_{T}(0)s$ and $\langle I_{\text{CHSH}}^{AC,T=1} \rangle \geq p_{T}(1)s$. The results are also displayed in \cref{fig:robustness}.

For all our numerical calculations, we consider the case where the only effect of the noise is that it decreases the observed Bell violation, i.e., there remains no noise in the flag register $T$. Specifically, we include the constraints in \cref{eq:projConstr1} in the numerical optimization, which guarantee that all parties observe the same flag each round. As argued below \cref{eq:projConstr1}, this noise model is justified since the flag corresponds to a classical signal in the honest implementation.    

\begin{figure}[h]
\includegraphics[width=9cm]{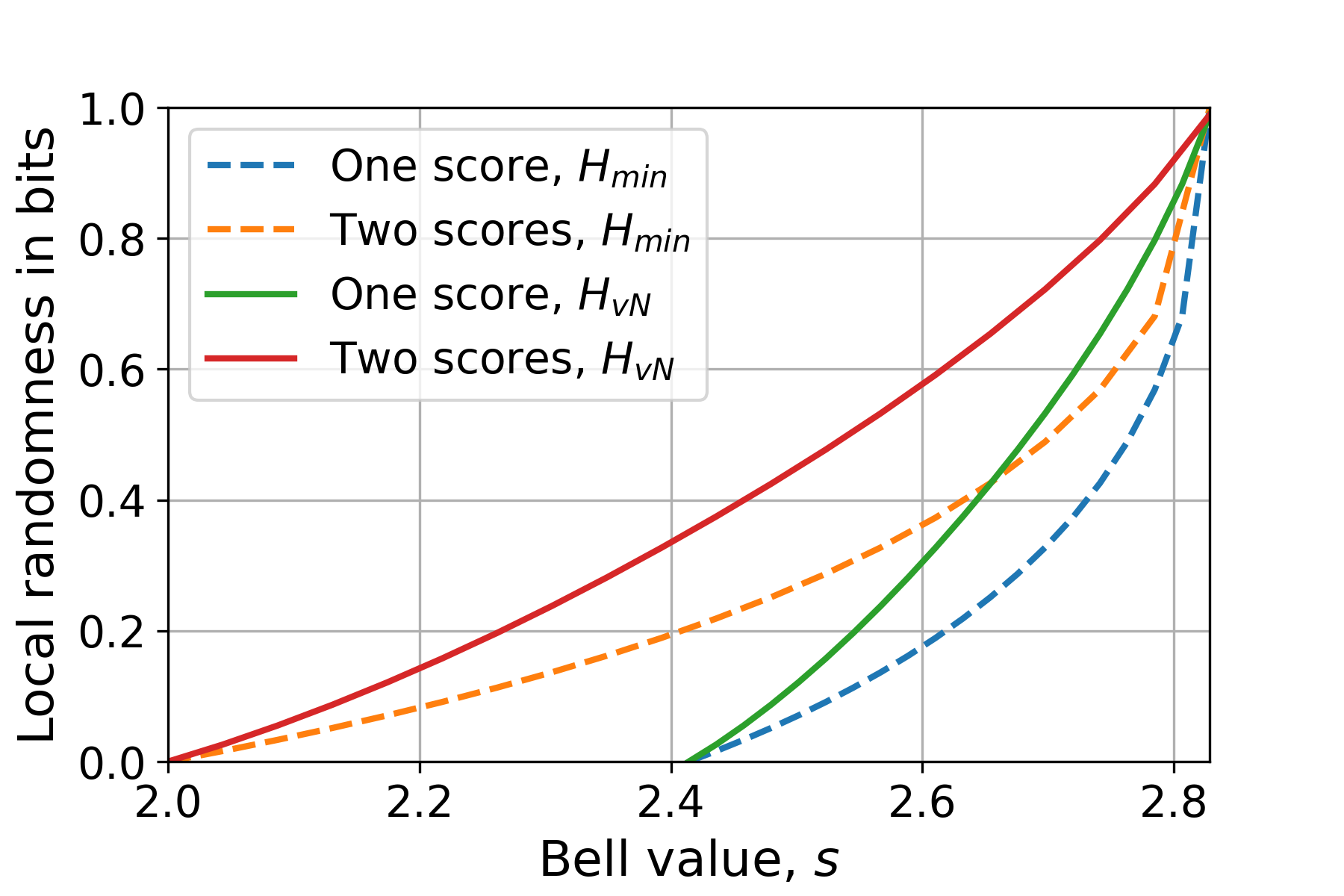}
\centering
\caption{Robustness of our construction to noise. We consider the local randomness of Alice's outcome $A$ when $T=0$, and set observed distribution of the flag to $p_{T}(0)= 1/2$. One score corresponds to the constraint $\langle I \rangle \geq s$, while two scores corresponds to the constraints $\langle I_{\text{CHSH}}^{AB,T=0} \rangle \geq p_{T}(0)s$ and $\langle I_{\text{CHSH}}^{AC,T=1} \rangle \geq p_{T}(1)s$. We also include all constraints in \cref{eq:projConstr1}. $H_{\text{min}}$ corresponds to lower bounds on the min-entropy obtained via the NPA hierarchy~\cite{NPA1,NPA2}, and $H_{\text{vN}}$ corresponds to lower bounds on the von Neumann entropy obtained using the technique of Ref.~\cite{Brown2024deviceindependent}. The data used in this figure is available at~\cite{dataset}.}
\label{fig:robustness}
\end{figure}

\section{Comparisons and adaptations} \label{sec:alt}

\subsection{Equivalent description without source randomness} \label{sec:alt_a}
We could equivalently describe the honest implementation of our protocol without the probabilistic selection of who shares entanglement with Alice. Let each party hold two qubits, $Q_{A_{1}}Q_{A_{2}}, Q_{B_{1}}Q_{B_{2}}, Q_{C_{1}}Q_{C_{2}}$. The honest state is given by
\begin{equation}
    \tilde{\rho} = \rho_{Q_{A_{1}}Q_{B_{1}}} \otimes \ketbra{+}{+}_{Q_{C_{1}}} \otimes \rho_{Q_{A_{2}}Q_{C_{2}}} \otimes \ketbra{+}{+}_{Q_{B_{2}}},
\end{equation}
where $\rho = \Phi_{0}$. The ideal observables are (interpreted as 4-outcome measurements)
\begin{equation}
    \begin{aligned}
        A_{0} &= \sigma_{Z}\otimes \sigma_{Z}, A_{1} = \sigma_{X} \otimes \sigma_{X}, \\
        B_{0/1} &= C_{0/1} = \Big[(\sigma_{Z} \pm \sigma_{X})/\sqrt{2} \Big]\otimes \Big[ (\sigma_{Z} \pm \sigma_{X})/\sqrt{2} \Big], \\
        B_{2} &= C_{2} = \sigma_{Z}\otimes \sigma_{Z}.
    \end{aligned}
\end{equation}
On generation rounds, Alice's outcome $a_{1}$ is perfectly correlated with $b_{1}$, and her second outcome $a_{2}$ is perfectly correlated with $c_{2}$. Moreover, the honest strategy maximally violates the Bell expression
\begin{equation}
    \langle I \rangle:= \langle \tilde{I}^{(1)}_{\mathrm{CHSH}} \rangle + \langle \tilde{I}^{(2)}_{\mathrm{CHSH}} \rangle, 
\end{equation}
where 
\begin{equation}
    \langle \tilde{I}^{(1)}_{\mathrm{CHSH}}\rangle  = \sum_{xya_{1}b_{1}}(-1)^{a_{1}+b_{1} + xy} p(a_{1}b_{1}|xy),
\end{equation}
and $p(a_{1}b_{1}|xy) = \sum_{a_{2}b_{2}c_{1}c_{2}}p(a_{1}a_{2}b_{1}b_{2}c_{1}c_{2}|xyz)$ is the marginal for Alice and Bob's first system. $\langle \tilde{I}^{(2)}_{\mathrm{CHSH}} \rangle $ is defined similarly between $A_{2}$ and $C_{2}$. We expect that the proof technique in \cref{app:proof} can be modified to show security of this protocol too, making this implementation a suitable candidate in practice, rather than the protocol in the main text which requires source randomness.  

This formulation also highlights the limitations of the ``per entangled state'' key rate measure. Every round of the protocol, an entangled state is distributed, and in the honest implementation one conference key bit is established. On the other hand, the protocol we first introduced also consumes one entangled state each round, but achieves a conference key rate of $1/2$, since only one party shares entanglement with Alice at a time. Yet, both protocols are identical in experimental cost, since they require two entangled pairs per conference key bit.

\subsection{Comparison to the device-dependent result of~\cite{Carrara21}} \label{sec:alt1}
Here we remark how our protocol differs from the device-dependent result of~\cite{Carrara21}. As briefly described in the introduction, in~\cite{Carrara21}, Alice generates an EPR pair, and randomly chooses who to send the other half to, whilst sending the other party a fixed un-entangled state. The parties then perform their measurements. For Bob, the probability he observes an error with respect to Alice's key, denoted as the error rate $Q_{B}$, is the probability he did not share entanglement and he guessed the wrong bit, which in the symmetric case leads to $Q_{B} = 1/4$. Similarly $Q_{C} = 1/4$. Bob and Carole then both try to correct their entire string to Alice's raw string, and the asymptotic rate is $1 - h_{2}(1/4) \approx 0.19$, where $h_{2}(\cdot)$ is the binary entropy. This is similar to the honest implementation of our protocol, except we include classical communication of the (untrusted) flag $T$, indicating who shared entanglement. Then, instead of Bob/Carole trying to correct rounds in which they shared no entanglement, they discard these, and distill a perfect key with Alice which is half the length of her total raw string. As a result, we achieve a rate of $1/2$. 

\subsection{Using other Bell expressions}

Another adaptation which could be made is to replace the CHSH test with another bipartite Bell inequality that certifies the presence of a maximally entangled state. The inequality tested between Alice and Bob could even differ from the one tested between Alice and Carole. For example, the three-parameter family of expressions presented in~\cite{Le23,Barizien24,WBC3} constitute all linear self-tests of the maximally entangled state with a single expression in the CHSH scenario, and provide additional degrees of freedom in the choice of measurements used in the Bell test. These Bell inequalities could immediately substitute the CHSH test in our protocol. It might also be possible to extend our results to higher dimensional systems. In~\cite{Farkas23}, the family of Bell expressions derived in~\cite{PereiraAlves22} were used for high dimensional DIQKD between two parties. Privacy follows from the quantum bound being achieved uniquely by the maximally entangled state in dimension $d$, along with a statement about Alice's key generating measurement~\cite{Farkas23}. These are the same ingredients used in the proof of \ifarxiv\cref{prop:ent}\else Proposition~1\fi, so it seems plausible that they could be used as the bipartite tests in our protocol, which could lead to high dimensional DICKA without GME.   


\begin{thebibliography}{58}%
\makeatletter
\providecommand \@ifxundefined [1]{%
 \@ifx{#1\undefined}
}%
\providecommand \@ifnum [1]{%
 \ifnum #1\expandafter \@firstoftwo
 \else \expandafter \@secondoftwo
 \fi
}%
\providecommand \@ifx [1]{%
 \ifx #1\expandafter \@firstoftwo
 \else \expandafter \@secondoftwo
 \fi
}%
\providecommand \natexlab [1]{#1}%
\providecommand \enquote  [1]{``#1''}%
\providecommand \bibnamefont  [1]{#1}%
\providecommand \bibfnamefont [1]{#1}%
\providecommand \citenamefont [1]{#1}%
\providecommand \href@noop [0]{\@secondoftwo}%
\providecommand \href [0]{\begingroup \@sanitize@url \@href}%
\providecommand \@href[1]{\@@startlink{#1}\@@href}%
\providecommand \@@href[1]{\endgroup#1\@@endlink}%
\providecommand \@sanitize@url [0]{\catcode `\\12\catcode `\$12\catcode
  `\&12\catcode `\#12\catcode `\^12\catcode `\_12\catcode `\%12\relax}%
\providecommand \@@startlink[1]{}%
\providecommand \@@endlink[0]{}%
\providecommand \url  [0]{\begingroup\@sanitize@url \@url }%
\providecommand \@url [1]{\endgroup\@href {#1}{\urlprefix }}%
\providecommand \urlprefix  [0]{URL }%
\providecommand \Eprint [0]{\href }%
\providecommand \doibase [0]{https://doi.org/}%
\providecommand \selectlanguage [0]{\@gobble}%
\providecommand \bibinfo  [0]{\@secondoftwo}%
\providecommand \bibfield  [0]{\@secondoftwo}%
\providecommand \translation [1]{[#1]}%
\providecommand \BibitemOpen [0]{}%
\providecommand \bibitemStop [0]{}%
\providecommand \bibitemNoStop [0]{.\EOS\space}%
\providecommand \EOS [0]{\spacefactor3000\relax}%
\providecommand \BibitemShut  [1]{\csname bibitem#1\endcsname}%
\let\auto@bib@innerbib\@empty
\bibitem [{\citenamefont {Cabello}(2000)}]{cabello00}%
  \BibitemOpen
  \bibfield  {author} {\bibinfo {author} {\bibfnamefont {A.}~\bibnamefont
  {Cabello}},\ }\href@noop {} {\bibinfo {title} {Multiparty key distribution
  and secret sharing based on entanglement swapping}} (\bibinfo {year}
  {2000}),\ \Eprint {https://arxiv.org/abs/quant-ph/0009025}
  {arXiv:quant-ph/0009025 [quant-ph]} \BibitemShut {NoStop}%
\bibitem [{\citenamefont {Chen}\ and\ \citenamefont {Lo}(2007)}]{Chen07}%
  \BibitemOpen
  \bibfield  {author} {\bibinfo {author} {\bibfnamefont {K.}~\bibnamefont
  {Chen}}\ and\ \bibinfo {author} {\bibfnamefont {H.}~\bibnamefont {Lo}},\
  }\bibfield  {title} {\bibinfo {title} {Multi-partite quantum cryptographic
  protocols with noisy {GHZ} states},\ }\href
  {https://doi.org/10.26421/QIC7.8-1} {\bibfield  {journal} {\bibinfo
  {journal} {Quantum Inf. Comput.}\ }\textbf {\bibinfo {volume} {7}},\ \bibinfo
  {pages} {689} (\bibinfo {year} {2007})}\BibitemShut {NoStop}%
\bibitem [{\citenamefont {Epping}\ \emph {et~al.}(2017)\citenamefont {Epping},
  \citenamefont {Kampermann}, \citenamefont {Macchiavello},\ and\ \citenamefont
  {Bruß}}]{Epping_2017}%
  \BibitemOpen
  \bibfield  {author} {\bibinfo {author} {\bibfnamefont {M.}~\bibnamefont
  {Epping}}, \bibinfo {author} {\bibfnamefont {H.}~\bibnamefont {Kampermann}},
  \bibinfo {author} {\bibfnamefont {C.}~\bibnamefont {Macchiavello}},\ and\
  \bibinfo {author} {\bibfnamefont {D.}~\bibnamefont {Bruß}},\ }\bibfield
  {title} {\bibinfo {title} {Multi-partite entanglement can speed up quantum
  key distribution in networks},\ }\href
  {https://doi.org/10.1088/1367-2630/aa8487} {\bibfield  {journal} {\bibinfo
  {journal} {New Journal of Physics}\ }\textbf {\bibinfo {volume} {19}},\
  \bibinfo {pages} {093012} (\bibinfo {year} {2017})}\BibitemShut {NoStop}%
\bibitem [{\citenamefont {Grasselli}\ \emph {et~al.}(2018)\citenamefont
  {Grasselli}, \citenamefont {Kampermann},\ and\ \citenamefont
  {Bruß}}]{Grasselli_2018}%
  \BibitemOpen
  \bibfield  {author} {\bibinfo {author} {\bibfnamefont {F.}~\bibnamefont
  {Grasselli}}, \bibinfo {author} {\bibfnamefont {H.}~\bibnamefont
  {Kampermann}},\ and\ \bibinfo {author} {\bibfnamefont {D.}~\bibnamefont
  {Bruß}},\ }\bibfield  {title} {\bibinfo {title} {Finite-key effects in
  multipartite quantum key distribution protocols},\ }\href
  {https://doi.org/10.1088/1367-2630/aaec34} {\bibfield  {journal} {\bibinfo
  {journal} {New Journal of Physics}\ }\textbf {\bibinfo {volume} {20}},\
  \bibinfo {pages} {113014} (\bibinfo {year} {2018})}\BibitemShut {NoStop}%
\bibitem [{\citenamefont {Grasselli}\ \emph {et~al.}(2019)\citenamefont
  {Grasselli}, \citenamefont {Kampermann},\ and\ \citenamefont
  {Bruß}}]{Grasselli_2019}%
  \BibitemOpen
  \bibfield  {author} {\bibinfo {author} {\bibfnamefont {F.}~\bibnamefont
  {Grasselli}}, \bibinfo {author} {\bibfnamefont {H.}~\bibnamefont
  {Kampermann}},\ and\ \bibinfo {author} {\bibfnamefont {D.}~\bibnamefont
  {Bruß}},\ }\bibfield  {title} {\bibinfo {title} {Conference key agreement
  with single-photon interference},\ }\href
  {https://doi.org/10.1088/1367-2630/ab573e} {\bibfield  {journal} {\bibinfo
  {journal} {New Journal of Physics}\ }\textbf {\bibinfo {volume} {21}},\
  \bibinfo {pages} {123002} (\bibinfo {year} {2019})}\BibitemShut {NoStop}%
\bibitem [{\citenamefont {Murta}\ \emph {et~al.}(2020)\citenamefont {Murta},
  \citenamefont {Grasselli}, \citenamefont {Kampermann},\ and\ \citenamefont
  {Bruß}}]{MurtaReview20}%
  \BibitemOpen
  \bibfield  {author} {\bibinfo {author} {\bibfnamefont {G.}~\bibnamefont
  {Murta}}, \bibinfo {author} {\bibfnamefont {F.}~\bibnamefont {Grasselli}},
  \bibinfo {author} {\bibfnamefont {H.}~\bibnamefont {Kampermann}},\ and\
  \bibinfo {author} {\bibfnamefont {D.}~\bibnamefont {Bruß}},\ }\bibfield
  {title} {\bibinfo {title} {Quantum conference key agreement: A review},\
  }\href {https://doi.org/https://doi.org/10.1002/qute.202000025} {\bibfield
  {journal} {\bibinfo  {journal} {Advanced Quantum Technologies}\ }\textbf
  {\bibinfo {volume} {3}},\ \bibinfo {pages} {2000025} (\bibinfo {year}
  {2020})}\BibitemShut {NoStop}%
\bibitem [{\citenamefont {Proietti}\ \emph {et~al.}(2021)\citenamefont
  {Proietti}, \citenamefont {Ho}, \citenamefont {Grasselli}, \citenamefont
  {Barrow}, \citenamefont {Malik},\ and\ \citenamefont
  {Fedrizzi}}]{Proietti_21}%
  \BibitemOpen
  \bibfield  {author} {\bibinfo {author} {\bibfnamefont {M.}~\bibnamefont
  {Proietti}}, \bibinfo {author} {\bibfnamefont {J.}~\bibnamefont {Ho}},
  \bibinfo {author} {\bibfnamefont {F.}~\bibnamefont {Grasselli}}, \bibinfo
  {author} {\bibfnamefont {P.}~\bibnamefont {Barrow}}, \bibinfo {author}
  {\bibfnamefont {M.}~\bibnamefont {Malik}},\ and\ \bibinfo {author}
  {\bibfnamefont {A.}~\bibnamefont {Fedrizzi}},\ }\bibfield  {title} {\bibinfo
  {title} {Experimental quantum conference key agreement},\ }\href
  {https://doi.org/10.1126/sciadv.abe0395} {\bibfield  {journal} {\bibinfo
  {journal} {Science Advances}\ }\textbf {\bibinfo {volume} {7}},\ \bibinfo
  {pages} {eabe0395} (\bibinfo {year} {2021})}\BibitemShut {NoStop}%
\bibitem [{\citenamefont {Carrara}\ \emph {et~al.}(2023)\citenamefont
  {Carrara}, \citenamefont {Murta},\ and\ \citenamefont
  {Grasselli}}]{Carrara23}%
  \BibitemOpen
  \bibfield  {author} {\bibinfo {author} {\bibfnamefont {G.}~\bibnamefont
  {Carrara}}, \bibinfo {author} {\bibfnamefont {G.}~\bibnamefont {Murta}},\
  and\ \bibinfo {author} {\bibfnamefont {F.}~\bibnamefont {Grasselli}},\
  }\bibfield  {title} {\bibinfo {title} {Overcoming fundamental bounds on
  quantum conference key agreement},\ }\href
  {https://doi.org/10.1103/PhysRevApplied.19.064017} {\bibfield  {journal}
  {\bibinfo  {journal} {Phys. Rev. Appl.}\ }\textbf {\bibinfo {volume} {19}},\
  \bibinfo {pages} {064017} (\bibinfo {year} {2023})}\BibitemShut {NoStop}%
\bibitem [{\citenamefont {Pickston}\ \emph {et~al.}(2023)\citenamefont
  {Pickston}, \citenamefont {Ho}, \citenamefont {Ulibarrena}, \citenamefont
  {Grasselli}, \citenamefont {Proietti}, \citenamefont {Morrison},
  \citenamefont {Barrow}, \citenamefont {Graffitti},\ and\ \citenamefont
  {Fedrizzi}}]{Pickston_2023}%
  \BibitemOpen
  \bibfield  {author} {\bibinfo {author} {\bibfnamefont {A.}~\bibnamefont
  {Pickston}}, \bibinfo {author} {\bibfnamefont {J.}~\bibnamefont {Ho}},
  \bibinfo {author} {\bibfnamefont {A.}~\bibnamefont {Ulibarrena}}, \bibinfo
  {author} {\bibfnamefont {F.}~\bibnamefont {Grasselli}}, \bibinfo {author}
  {\bibfnamefont {M.}~\bibnamefont {Proietti}}, \bibinfo {author}
  {\bibfnamefont {C.~L.}\ \bibnamefont {Morrison}}, \bibinfo {author}
  {\bibfnamefont {P.}~\bibnamefont {Barrow}}, \bibinfo {author} {\bibfnamefont
  {F.}~\bibnamefont {Graffitti}},\ and\ \bibinfo {author} {\bibfnamefont
  {A.}~\bibnamefont {Fedrizzi}},\ }\bibfield  {title} {\bibinfo {title}
  {Conference key agreement in a quantum network},\ }\href
  {https://doi.org/10.1038/s41534-023-00750-4} {\bibfield  {journal} {\bibinfo
  {journal} {npj Quantum Information}\ }\textbf {\bibinfo {volume} {9}},\
  \bibinfo {pages} {82} (\bibinfo {year} {2023})}\BibitemShut {NoStop}%
\bibitem [{\citenamefont {Ribeiro}\ \emph {et~al.}(2018)\citenamefont
  {Ribeiro}, \citenamefont {Murta},\ and\ \citenamefont {Wehner}}]{Ribeiro18}%
  \BibitemOpen
  \bibfield  {author} {\bibinfo {author} {\bibfnamefont {J.}~\bibnamefont
  {Ribeiro}}, \bibinfo {author} {\bibfnamefont {G.}~\bibnamefont {Murta}},\
  and\ \bibinfo {author} {\bibfnamefont {S.}~\bibnamefont {Wehner}},\
  }\bibfield  {title} {\bibinfo {title} {Fully device-independent conference
  key agreement},\ }\href {https://doi.org/10.1103/PhysRevA.97.022307}
  {\bibfield  {journal} {\bibinfo  {journal} {Phys. Rev. A}\ }\textbf {\bibinfo
  {volume} {97}},\ \bibinfo {pages} {022307} (\bibinfo {year}
  {2018})}\BibitemShut {NoStop}%
\bibitem [{\citenamefont {Holz}\ \emph {et~al.}(2019)\citenamefont {Holz},
  \citenamefont {Miller}, \citenamefont {Kampermann},\ and\ \citenamefont
  {Bru\ss{}}}]{Holz19}%
  \BibitemOpen
  \bibfield  {author} {\bibinfo {author} {\bibfnamefont {T.}~\bibnamefont
  {Holz}}, \bibinfo {author} {\bibfnamefont {D.}~\bibnamefont {Miller}},
  \bibinfo {author} {\bibfnamefont {H.}~\bibnamefont {Kampermann}},\ and\
  \bibinfo {author} {\bibfnamefont {D.}~\bibnamefont {Bru\ss{}}},\ }\bibfield
  {title} {\bibinfo {title} {Comment on ``fully device-independent conference
  key agreement''},\ }\href {https://doi.org/10.1103/PhysRevA.100.026301}
  {\bibfield  {journal} {\bibinfo  {journal} {Phys. Rev. A}\ }\textbf {\bibinfo
  {volume} {100}},\ \bibinfo {pages} {026301} (\bibinfo {year}
  {2019})}\BibitemShut {NoStop}%
\bibitem [{\citenamefont {Ribeiro}\ \emph {et~al.}(2019)\citenamefont
  {Ribeiro}, \citenamefont {Murta},\ and\ \citenamefont {Wehner}}]{Ribeiro19}%
  \BibitemOpen
  \bibfield  {author} {\bibinfo {author} {\bibfnamefont {J.}~\bibnamefont
  {Ribeiro}}, \bibinfo {author} {\bibfnamefont {G.}~\bibnamefont {Murta}},\
  and\ \bibinfo {author} {\bibfnamefont {S.}~\bibnamefont {Wehner}},\
  }\bibfield  {title} {\bibinfo {title} {Reply to ``comment on `fully
  device-independent conference key agreement' ''},\ }\href
  {https://doi.org/10.1103/PhysRevA.100.026302} {\bibfield  {journal} {\bibinfo
   {journal} {Phys. Rev. A}\ }\textbf {\bibinfo {volume} {100}},\ \bibinfo
  {pages} {026302} (\bibinfo {year} {2019})}\BibitemShut {NoStop}%
\bibitem [{\citenamefont {Holz}\ \emph {et~al.}(2020)\citenamefont {Holz},
  \citenamefont {Kampermann},\ and\ \citenamefont {Bru\ss{}}}]{Holz20}%
  \BibitemOpen
  \bibfield  {author} {\bibinfo {author} {\bibfnamefont {T.}~\bibnamefont
  {Holz}}, \bibinfo {author} {\bibfnamefont {H.}~\bibnamefont {Kampermann}},\
  and\ \bibinfo {author} {\bibfnamefont {D.}~\bibnamefont {Bru\ss{}}},\
  }\bibfield  {title} {\bibinfo {title} {Genuine multipartite {B}ell inequality
  for device-independent conference key agreement},\ }\href
  {https://doi.org/10.1103/PhysRevResearch.2.023251} {\bibfield  {journal}
  {\bibinfo  {journal} {Phys. Rev. Res.}\ }\textbf {\bibinfo {volume} {2}},\
  \bibinfo {pages} {023251} (\bibinfo {year} {2020})}\BibitemShut {NoStop}%
\bibitem [{\citenamefont {Einstein}\ \emph {et~al.}(1935)\citenamefont
  {Einstein}, \citenamefont {Podolsky},\ and\ \citenamefont {Rosen}}]{EPR}%
  \BibitemOpen
  \bibfield  {author} {\bibinfo {author} {\bibfnamefont {A.}~\bibnamefont
  {Einstein}}, \bibinfo {author} {\bibfnamefont {B.}~\bibnamefont {Podolsky}},\
  and\ \bibinfo {author} {\bibfnamefont {N.}~\bibnamefont {Rosen}},\ }\bibfield
   {title} {\bibinfo {title} {Can quantum-mechanical description of physical
  reality be considered complete?},\ }\href
  {https://doi.org/10.1103/PhysRev.47.777} {\bibfield  {journal} {\bibinfo
  {journal} {Physical Review}\ }\textbf {\bibinfo {volume} {47}},\ \bibinfo
  {pages} {777} (\bibinfo {year} {1935})}\BibitemShut {NoStop}%
\bibitem [{\citenamefont {Bell}(1987)}]{Bell_book}%
  \BibitemOpen
  \bibfield  {author} {\bibinfo {author} {\bibfnamefont {J.~S.}\ \bibnamefont
  {Bell}},\ }\href {https://doi.org/10.1017/CBO9780511815676} {\emph {\bibinfo
  {title} {Speakable and unspeakable in quantum mechanics}}}\ (\bibinfo
  {publisher} {Cambridge University Press},\ \bibinfo {year}
  {1987})\BibitemShut {NoStop}%
\bibitem [{\citenamefont {Barrett}\ \emph
  {et~al.}(2005{\natexlab{a}})\citenamefont {Barrett}, \citenamefont {Linden},
  \citenamefont {Massar}, \citenamefont {Pironio}, \citenamefont {Popescu},\
  and\ \citenamefont {Roberts}}]{BarrettNonlocalResource}%
  \BibitemOpen
  \bibfield  {author} {\bibinfo {author} {\bibfnamefont {J.}~\bibnamefont
  {Barrett}}, \bibinfo {author} {\bibfnamefont {N.}~\bibnamefont {Linden}},
  \bibinfo {author} {\bibfnamefont {S.}~\bibnamefont {Massar}}, \bibinfo
  {author} {\bibfnamefont {S.}~\bibnamefont {Pironio}}, \bibinfo {author}
  {\bibfnamefont {S.}~\bibnamefont {Popescu}},\ and\ \bibinfo {author}
  {\bibfnamefont {D.}~\bibnamefont {Roberts}},\ }\bibfield  {title} {\bibinfo
  {title} {Nonlocal correlations as an information-theoretic resource},\ }\href
  {https://doi.org/10.1103/PhysRevA.71.022101} {\bibfield  {journal} {\bibinfo
  {journal} {Physical Review A}\ }\textbf {\bibinfo {volume} {71}},\ \bibinfo
  {pages} {022101} (\bibinfo {year} {2005}{\natexlab{a}})}\BibitemShut
  {NoStop}%
\bibitem [{\citenamefont {Barrett}\ \emph
  {et~al.}(2005{\natexlab{b}})\citenamefont {Barrett}, \citenamefont {Hardy},\
  and\ \citenamefont {Kent}}]{BHK}%
  \BibitemOpen
  \bibfield  {author} {\bibinfo {author} {\bibfnamefont {J.}~\bibnamefont
  {Barrett}}, \bibinfo {author} {\bibfnamefont {L.}~\bibnamefont {Hardy}},\
  and\ \bibinfo {author} {\bibfnamefont {A.}~\bibnamefont {Kent}},\ }\bibfield
  {title} {\bibinfo {title} {No signalling and quantum key distribution},\
  }\href {https://doi.org/10.1103/PhysRevLett.95.010503} {\bibfield  {journal}
  {\bibinfo  {journal} {Physical Review Letters}\ }\textbf {\bibinfo {volume}
  {95}},\ \bibinfo {pages} {010503} (\bibinfo {year}
  {2005}{\natexlab{b}})}\BibitemShut {NoStop}%
\bibitem [{\citenamefont {Ac\'in}\ \emph {et~al.}(2007)\citenamefont {Ac\'in},
  \citenamefont {Brunner}, \citenamefont {Gisin}, \citenamefont {Massar},
  \citenamefont {Pironio},\ and\ \citenamefont {Scarani}}]{ABGMPS}%
  \BibitemOpen
  \bibfield  {author} {\bibinfo {author} {\bibfnamefont {A.}~\bibnamefont
  {Ac\'in}}, \bibinfo {author} {\bibfnamefont {N.}~\bibnamefont {Brunner}},
  \bibinfo {author} {\bibfnamefont {N.}~\bibnamefont {Gisin}}, \bibinfo
  {author} {\bibfnamefont {S.}~\bibnamefont {Massar}}, \bibinfo {author}
  {\bibfnamefont {S.}~\bibnamefont {Pironio}},\ and\ \bibinfo {author}
  {\bibfnamefont {V.}~\bibnamefont {Scarani}},\ }\bibfield  {title} {\bibinfo
  {title} {Device-independent security of quantum cryptography against
  collective attacks},\ }\href {https://doi.org/10.1103/PhysRevLett.98.230501}
  {\bibfield  {journal} {\bibinfo  {journal} {Physical Review Letters}\
  }\textbf {\bibinfo {volume} {98}},\ \bibinfo {pages} {230501} (\bibinfo
  {year} {2007})}\BibitemShut {NoStop}%
\bibitem [{\citenamefont {Pironio}\ \emph {et~al.}(2009)\citenamefont
  {Pironio}, \citenamefont {Acin}, \citenamefont {Brunner}, \citenamefont
  {Gisin}, \citenamefont {Massar},\ and\ \citenamefont {Scarani}}]{PABGMS}%
  \BibitemOpen
  \bibfield  {author} {\bibinfo {author} {\bibfnamefont {S.}~\bibnamefont
  {Pironio}}, \bibinfo {author} {\bibfnamefont {A.}~\bibnamefont {Acin}},
  \bibinfo {author} {\bibfnamefont {N.}~\bibnamefont {Brunner}}, \bibinfo
  {author} {\bibfnamefont {N.}~\bibnamefont {Gisin}}, \bibinfo {author}
  {\bibfnamefont {S.}~\bibnamefont {Massar}},\ and\ \bibinfo {author}
  {\bibfnamefont {V.}~\bibnamefont {Scarani}},\ }\bibfield  {title} {\bibinfo
  {title} {Device-independent quantum key distribution secure against
  collective attacks},\ }\href {https://doi.org/10.1088/1367-2630/11/4/045021}
  {\bibfield  {journal} {\bibinfo  {journal} {New Journal of Physics}\ }\textbf
  {\bibinfo {volume} {11}},\ \bibinfo {pages} {045021} (\bibinfo {year}
  {2009})}\BibitemShut {NoStop}%
\bibitem [{\citenamefont {Barrett}\ \emph {et~al.}(2012)\citenamefont
  {Barrett}, \citenamefont {Colbeck},\ and\ \citenamefont {Kent}}]{bcktwo}%
  \BibitemOpen
  \bibfield  {author} {\bibinfo {author} {\bibfnamefont {J.}~\bibnamefont
  {Barrett}}, \bibinfo {author} {\bibfnamefont {R.}~\bibnamefont {Colbeck}},\
  and\ \bibinfo {author} {\bibfnamefont {A.}~\bibnamefont {Kent}},\ }\bibfield
  {title} {\bibinfo {title} {Unconditionally secure device-independent quantum
  key distribution with only two devices},\ }\href
  {https://doi.org/10.1103/PhysRevA.86.062326} {\bibfield  {journal} {\bibinfo
  {journal} {Physical Review A}\ }\textbf {\bibinfo {volume} {86}},\ \bibinfo
  {pages} {062326} (\bibinfo {year} {2012})}\BibitemShut {NoStop}%
\bibitem [{\citenamefont {Vazirani}\ and\ \citenamefont {Vidick}(2014)}]{VV2}%
  \BibitemOpen
  \bibfield  {author} {\bibinfo {author} {\bibfnamefont {U.}~\bibnamefont
  {Vazirani}}\ and\ \bibinfo {author} {\bibfnamefont {T.}~\bibnamefont
  {Vidick}},\ }\bibfield  {title} {\bibinfo {title} {Fully device-independent
  quantum key distribution},\ }\href
  {https://doi.org/10.1103/PhysRevLett.113.140501} {\bibfield  {journal}
  {\bibinfo  {journal} {Physical Review Letters}\ }\textbf {\bibinfo {volume}
  {113}},\ \bibinfo {pages} {140501} (\bibinfo {year} {2014})}\BibitemShut
  {NoStop}%
\bibitem [{\citenamefont {Grasselli}\ \emph {et~al.}(2023)\citenamefont
  {Grasselli}, \citenamefont {Murta}, \citenamefont {Kampermann},\ and\
  \citenamefont {Bruß}}]{Grasselli_2023}%
  \BibitemOpen
  \bibfield  {author} {\bibinfo {author} {\bibfnamefont {F.}~\bibnamefont
  {Grasselli}}, \bibinfo {author} {\bibfnamefont {G.}~\bibnamefont {Murta}},
  \bibinfo {author} {\bibfnamefont {H.}~\bibnamefont {Kampermann}},\ and\
  \bibinfo {author} {\bibfnamefont {D.}~\bibnamefont {Bruß}},\ }\bibfield
  {title} {\bibinfo {title} {Boosting device-independent cryptography with
  tripartite nonlocality},\ }\href {https://doi.org/10.22331/q-2023-04-13-980}
  {\bibfield  {journal} {\bibinfo  {journal} {Quantum}\ }\textbf {\bibinfo
  {volume} {7}},\ \bibinfo {pages} {980} (\bibinfo {year} {2023})}\BibitemShut
  {NoStop}%
\bibitem [{\citenamefont {Grasselli}\ \emph {et~al.}(2021)\citenamefont
  {Grasselli}, \citenamefont {Murta}, \citenamefont {Kampermann},\ and\
  \citenamefont {Bru\ss{}}}]{GrasselliMulti21}%
  \BibitemOpen
  \bibfield  {author} {\bibinfo {author} {\bibfnamefont {F.}~\bibnamefont
  {Grasselli}}, \bibinfo {author} {\bibfnamefont {G.}~\bibnamefont {Murta}},
  \bibinfo {author} {\bibfnamefont {H.}~\bibnamefont {Kampermann}},\ and\
  \bibinfo {author} {\bibfnamefont {D.}~\bibnamefont {Bru\ss{}}},\ }\bibfield
  {title} {\bibinfo {title} {Entropy bounds for multiparty device-independent
  cryptography},\ }\href {https://doi.org/10.1103/PRXQuantum.2.010308}
  {\bibfield  {journal} {\bibinfo  {journal} {PRX Quantum}\ }\textbf {\bibinfo
  {volume} {2}},\ \bibinfo {pages} {010308} (\bibinfo {year}
  {2021})}\BibitemShut {NoStop}%
\bibitem [{\citenamefont {Woodhead}\ \emph {et~al.}(2018)\citenamefont
  {Woodhead}, \citenamefont {Bourdoncle},\ and\ \citenamefont
  {Acín}}]{Woodhead_2018}%
  \BibitemOpen
  \bibfield  {author} {\bibinfo {author} {\bibfnamefont {E.}~\bibnamefont
  {Woodhead}}, \bibinfo {author} {\bibfnamefont {B.}~\bibnamefont
  {Bourdoncle}},\ and\ \bibinfo {author} {\bibfnamefont {A.}~\bibnamefont
  {Acín}},\ }\bibfield  {title} {\bibinfo {title} {Randomness versus
  nonlocality in the {M}ermin-{B}ell experiment with three parties},\ }\href
  {https://doi.org/10.22331/q-2018-08-17-82} {\bibfield  {journal} {\bibinfo
  {journal} {Quantum}\ }\textbf {\bibinfo {volume} {2}},\ \bibinfo {pages} {82}
  (\bibinfo {year} {2018})}\BibitemShut {NoStop}%
\bibitem [{\citenamefont {Horodecki}\ \emph {et~al.}(2022)\citenamefont
  {Horodecki}, \citenamefont {Winczewski},\ and\ \citenamefont
  {Das}}]{Horodecki22}%
  \BibitemOpen
  \bibfield  {author} {\bibinfo {author} {\bibfnamefont {K.}~\bibnamefont
  {Horodecki}}, \bibinfo {author} {\bibfnamefont {M.}~\bibnamefont
  {Winczewski}},\ and\ \bibinfo {author} {\bibfnamefont {S.}~\bibnamefont
  {Das}},\ }\bibfield  {title} {\bibinfo {title} {Fundamental limitations on
  the device-independent quantum conference key agreement},\ }\href
  {https://doi.org/10.1103/PhysRevA.105.022604} {\bibfield  {journal} {\bibinfo
   {journal} {Phys. Rev. A}\ }\textbf {\bibinfo {volume} {105}},\ \bibinfo
  {pages} {022604} (\bibinfo {year} {2022})}\BibitemShut {NoStop}%
\bibitem [{\citenamefont {Philip}\ \emph {et~al.}(2023)\citenamefont {Philip},
  \citenamefont {Kaur}, \citenamefont {Bierhorst},\ and\ \citenamefont
  {Wilde}}]{Philip23}%
  \BibitemOpen
  \bibfield  {author} {\bibinfo {author} {\bibfnamefont {A.}~\bibnamefont
  {Philip}}, \bibinfo {author} {\bibfnamefont {E.}~\bibnamefont {Kaur}},
  \bibinfo {author} {\bibfnamefont {P.}~\bibnamefont {Bierhorst}},\ and\
  \bibinfo {author} {\bibfnamefont {M.~M.}\ \bibnamefont {Wilde}},\ }\bibfield
  {title} {\bibinfo {title} {Multipartite intrinsic non-locality and
  device-independent conference key agreement},\ }\href
  {https://doi.org/10.22331/q-2023-01-19-898} {\bibfield  {journal} {\bibinfo
  {journal} {{Quantum}}\ }\textbf {\bibinfo {volume} {7}},\ \bibinfo {pages}
  {898} (\bibinfo {year} {2023})}\BibitemShut {NoStop}%
\bibitem [{\citenamefont {Barrett}\ \emph {et~al.}(2013)\citenamefont
  {Barrett}, \citenamefont {Colbeck},\ and\ \citenamefont {Kent}}]{bckone}%
  \BibitemOpen
  \bibfield  {author} {\bibinfo {author} {\bibfnamefont {J.}~\bibnamefont
  {Barrett}}, \bibinfo {author} {\bibfnamefont {R.}~\bibnamefont {Colbeck}},\
  and\ \bibinfo {author} {\bibfnamefont {A.}~\bibnamefont {Kent}},\ }\bibfield
  {title} {\bibinfo {title} {Memory attacks on device-independent quantum
  cryptography},\ }\href {https://doi.org/10.1103/PhysRevLett.110.010503}
  {\bibfield  {journal} {\bibinfo  {journal} {Physical Review Letters}\
  }\textbf {\bibinfo {volume} {106}},\ \bibinfo {pages} {010503} (\bibinfo
  {year} {2013})}\BibitemShut {NoStop}%
\bibitem [{\citenamefont {Carrara}\ \emph {et~al.}(2021)\citenamefont
  {Carrara}, \citenamefont {Kampermann}, \citenamefont {Bru\ss{}},\ and\
  \citenamefont {Murta}}]{Carrara21}%
  \BibitemOpen
  \bibfield  {author} {\bibinfo {author} {\bibfnamefont {G.}~\bibnamefont
  {Carrara}}, \bibinfo {author} {\bibfnamefont {H.}~\bibnamefont {Kampermann}},
  \bibinfo {author} {\bibfnamefont {D.}~\bibnamefont {Bru\ss{}}},\ and\
  \bibinfo {author} {\bibfnamefont {G.}~\bibnamefont {Murta}},\ }\bibfield
  {title} {\bibinfo {title} {Genuine multipartite entanglement is not a
  precondition for secure conference key agreement},\ }\href
  {https://doi.org/10.1103/PhysRevResearch.3.013264} {\bibfield  {journal}
  {\bibinfo  {journal} {Phys. Rev. Res.}\ }\textbf {\bibinfo {volume} {3}},\
  \bibinfo {pages} {013264} (\bibinfo {year} {2021})}\BibitemShut {NoStop}%
\bibitem [{\citenamefont {{Clauser}}\ \emph {et~al.}(1969)\citenamefont
  {{Clauser}}, \citenamefont {{Horne}}, \citenamefont {{Shimony}},\ and\
  \citenamefont {{Holt}}}]{CHSH}%
  \BibitemOpen
  \bibfield  {author} {\bibinfo {author} {\bibfnamefont {J.~F.}\ \bibnamefont
  {{Clauser}}}, \bibinfo {author} {\bibfnamefont {M.~A.}\ \bibnamefont
  {{Horne}}}, \bibinfo {author} {\bibfnamefont {A.}~\bibnamefont {{Shimony}}},\
  and\ \bibinfo {author} {\bibfnamefont {R.~A.}\ \bibnamefont {{Holt}}},\
  }\bibfield  {title} {\bibinfo {title} {Proposed experiment to test local
  hidden-variable theories},\ }\href
  {https://doi.org/10.1103/PhysRevLett.23.880} {\bibfield  {journal} {\bibinfo
  {journal} {Physical Review Letters}\ }\textbf {\bibinfo {volume} {23}},\
  \bibinfo {pages} {880} (\bibinfo {year} {1969})}\BibitemShut {NoStop}%
\bibitem [{\citenamefont {Dupuis}\ \emph {et~al.}(2020)\citenamefont {Dupuis},
  \citenamefont {Fawzi},\ and\ \citenamefont {Renner}}]{DFR}%
  \BibitemOpen
  \bibfield  {author} {\bibinfo {author} {\bibfnamefont {F.}~\bibnamefont
  {Dupuis}}, \bibinfo {author} {\bibfnamefont {O.}~\bibnamefont {Fawzi}},\ and\
  \bibinfo {author} {\bibfnamefont {R.}~\bibnamefont {Renner}},\ }\bibfield
  {title} {\bibinfo {title} {Entropy accumulation},\ }\href
  {https://doi.org/10.1007/s00220-020-03839-5} {\bibfield  {journal} {\bibinfo
  {journal} {Communication in Mathematical Physics}\ }\textbf {\bibinfo
  {volume} {379}},\ \bibinfo {pages} {867} (\bibinfo {year}
  {2020})}\BibitemShut {NoStop}%
\bibitem [{\citenamefont {Arnon-Friedman}\ \emph {et~al.}(2018)\citenamefont
  {Arnon-Friedman}, \citenamefont {Dupuis}, \citenamefont {Fawzi},
  \citenamefont {Renner},\ and\ \citenamefont {Vidick}}]{ADFRV}%
  \BibitemOpen
  \bibfield  {author} {\bibinfo {author} {\bibfnamefont {R.}~\bibnamefont
  {Arnon-Friedman}}, \bibinfo {author} {\bibfnamefont {F.}~\bibnamefont
  {Dupuis}}, \bibinfo {author} {\bibfnamefont {O.}~\bibnamefont {Fawzi}},
  \bibinfo {author} {\bibfnamefont {R.}~\bibnamefont {Renner}},\ and\ \bibinfo
  {author} {\bibfnamefont {T.}~\bibnamefont {Vidick}},\ }\bibfield  {title}
  {\bibinfo {title} {Practical device-independent quantum cryptography via
  entropy accumulation},\ }\href {https://doi.org/10.1038/s41467-017-02307-4}
  {\bibfield  {journal} {\bibinfo  {journal} {Nature communications}\ }\textbf
  {\bibinfo {volume} {9}},\ \bibinfo {pages} {459} (\bibinfo {year}
  {2018})}\BibitemShut {NoStop}%
\bibitem [{\citenamefont {Plenio}\ and\ \citenamefont
  {Virmani}(2007)}]{Plenio07}%
  \BibitemOpen
  \bibfield  {author} {\bibinfo {author} {\bibfnamefont {M.~B.}\ \bibnamefont
  {Plenio}}\ and\ \bibinfo {author} {\bibfnamefont {S.~S.}\ \bibnamefont
  {Virmani}},\ }\bibfield  {title} {\bibinfo {title} {{An Introduction to
  Entanglement Theory}},\ }\href {https://doi.org/10.1007/978-3-319-04063-9_8}
  {\bibfield  {journal} {\bibinfo  {journal} {Quant. Inf. Comput.}\ }\textbf
  {\bibinfo {volume} {7}},\ \bibinfo {pages} {001} (\bibinfo {year}
  {2007})}\BibitemShut {NoStop}%
\bibitem [{\citenamefont {Arnon-Friedman}\ and\ \citenamefont
  {Bancal}(2019)}]{Arnon_Friedman_2019}%
  \BibitemOpen
  \bibfield  {author} {\bibinfo {author} {\bibfnamefont {R.}~\bibnamefont
  {Arnon-Friedman}}\ and\ \bibinfo {author} {\bibfnamefont {J.-D.}\
  \bibnamefont {Bancal}},\ }\bibfield  {title} {\bibinfo {title}
  {Device-independent certification of one-shot distillable entanglement},\
  }\href {https://doi.org/10.1088/1367-2630/aafef6} {\bibfield  {journal}
  {\bibinfo  {journal} {New Journal of Physics}\ }\textbf {\bibinfo {volume}
  {21}},\ \bibinfo {pages} {033010} (\bibinfo {year} {2019})}\BibitemShut
  {NoStop}%
\bibitem [{\citenamefont {Philip}\ and\ \citenamefont
  {Wilde}(2025)}]{Philip25}%
  \BibitemOpen
  \bibfield  {author} {\bibinfo {author} {\bibfnamefont {A.}~\bibnamefont
  {Philip}}\ and\ \bibinfo {author} {\bibfnamefont {M.~M.}\ \bibnamefont
  {Wilde}},\ }\bibfield  {title} {\bibinfo {title} {Device-independent
  certification of multipartite distillable entanglement},\ }\href
  {https://doi.org/10.1103/PhysRevA.111.012436} {\bibfield  {journal} {\bibinfo
   {journal} {Phys. Rev. A}\ }\textbf {\bibinfo {volume} {111}},\ \bibinfo
  {pages} {012436} (\bibinfo {year} {2025})}\BibitemShut {NoStop}%
\bibitem [{\citenamefont {Hayashi}\ and\ \citenamefont
  {Tsurumaru}(2016)}]{Hayashi_2016}%
  \BibitemOpen
  \bibfield  {author} {\bibinfo {author} {\bibfnamefont {M.}~\bibnamefont
  {Hayashi}}\ and\ \bibinfo {author} {\bibfnamefont {T.}~\bibnamefont
  {Tsurumaru}},\ }\bibfield  {title} {\bibinfo {title} {More efficient privacy
  amplification with less random seeds via dual universal hash function},\
  }\href {https://doi.org/10.1109/tit.2016.2526018} {\bibfield  {journal}
  {\bibinfo  {journal} {IEEE Transactions on Information Theory}\ }\textbf
  {\bibinfo {volume} {62}},\ \bibinfo {pages} {2213–2232} (\bibinfo {year}
  {2016})}\BibitemShut {NoStop}%
\bibitem [{\citenamefont {Cao}\ \emph {et~al.}(2021{\natexlab{a}})\citenamefont
  {Cao}, \citenamefont {Lu}, \citenamefont {Li}, \citenamefont {Gu},
  \citenamefont {Yin},\ and\ \citenamefont {Chen}}]{Cao21}%
  \BibitemOpen
  \bibfield  {author} {\bibinfo {author} {\bibfnamefont {X.-Y.}\ \bibnamefont
  {Cao}}, \bibinfo {author} {\bibfnamefont {Y.-S.}\ \bibnamefont {Lu}},
  \bibinfo {author} {\bibfnamefont {Z.}~\bibnamefont {Li}}, \bibinfo {author}
  {\bibfnamefont {J.}~\bibnamefont {Gu}}, \bibinfo {author} {\bibfnamefont
  {H.-L.}\ \bibnamefont {Yin}},\ and\ \bibinfo {author} {\bibfnamefont {Z.-B.}\
  \bibnamefont {Chen}},\ }\bibfield  {title} {\bibinfo {title} {High key rate
  quantum conference key agreement with unconditional security},\ }\href
  {https://doi.org/10.1109/ACCESS.2021.3113939} {\bibfield  {journal} {\bibinfo
   {journal} {IEEE Access}\ }\textbf {\bibinfo {volume} {9}},\ \bibinfo {pages}
  {128870} (\bibinfo {year} {2021}{\natexlab{a}})}\BibitemShut {NoStop}%
\bibitem [{\citenamefont {Cao}\ \emph {et~al.}(2021{\natexlab{b}})\citenamefont
  {Cao}, \citenamefont {Gu}, \citenamefont {Lu}, \citenamefont {Yin},\ and\
  \citenamefont {Chen}}]{Cao_2021b}%
  \BibitemOpen
  \bibfield  {author} {\bibinfo {author} {\bibfnamefont {X.-Y.}\ \bibnamefont
  {Cao}}, \bibinfo {author} {\bibfnamefont {J.}~\bibnamefont {Gu}}, \bibinfo
  {author} {\bibfnamefont {Y.-S.}\ \bibnamefont {Lu}}, \bibinfo {author}
  {\bibfnamefont {H.-L.}\ \bibnamefont {Yin}},\ and\ \bibinfo {author}
  {\bibfnamefont {Z.-B.}\ \bibnamefont {Chen}},\ }\bibfield  {title} {\bibinfo
  {title} {Coherent one-way quantum conference key agreement based on twin
  field},\ }\href {https://doi.org/10.1088/1367-2630/abef98} {\bibfield
  {journal} {\bibinfo  {journal} {New Journal of Physics}\ }\textbf {\bibinfo
  {volume} {23}},\ \bibinfo {pages} {043002} (\bibinfo {year}
  {2021}{\natexlab{b}})}\BibitemShut {NoStop}%
\bibitem [{\citenamefont {Wooltorton}\ \emph {et~al.}(2023)\citenamefont
  {Wooltorton}, \citenamefont {Brown},\ and\ \citenamefont {Colbeck}}]{WBC2}%
  \BibitemOpen
  \bibfield  {author} {\bibinfo {author} {\bibfnamefont {L.}~\bibnamefont
  {Wooltorton}}, \bibinfo {author} {\bibfnamefont {P.}~\bibnamefont {Brown}},\
  and\ \bibinfo {author} {\bibfnamefont {R.}~\bibnamefont {Colbeck}},\ }\href
  {https://arxiv.org/abs/2308.07030} {\bibinfo {title} {Expanding bipartite
  {B}ell inequalities for maximum multi-partite randomness}} (\bibinfo {year}
  {2023}),\ \Eprint {https://arxiv.org/abs/2308.07030} {arXiv:2308.07030
  [quant-ph]} \BibitemShut {NoStop}%
\bibitem [{\citenamefont {Aspect}\ \emph {et~al.}(1982)\citenamefont {Aspect},
  \citenamefont {Dalibard},\ and\ \citenamefont {Roger}}]{Aspect82}%
  \BibitemOpen
  \bibfield  {author} {\bibinfo {author} {\bibfnamefont {A.}~\bibnamefont
  {Aspect}}, \bibinfo {author} {\bibfnamefont {J.}~\bibnamefont {Dalibard}},\
  and\ \bibinfo {author} {\bibfnamefont {G.}~\bibnamefont {Roger}},\ }\bibfield
   {title} {\bibinfo {title} {Experimental test of {B}ell's inequalities using
  time-varying analyzers},\ }\href
  {https://doi.org/10.1103/PhysRevLett.49.1804} {\bibfield  {journal} {\bibinfo
   {journal} {Phys. Rev. Lett.}\ }\textbf {\bibinfo {volume} {49}},\ \bibinfo
  {pages} {1804} (\bibinfo {year} {1982})}\BibitemShut {NoStop}%
\bibitem [{\citenamefont {Hensen}\ \emph {et~al.}(2015)\citenamefont {Hensen},
  \citenamefont {Bernien}, \citenamefont {Dr{\'e}au}, \citenamefont {Reiserer},
  \citenamefont {Kalb}, \citenamefont {Blok}, \citenamefont {Ruitenberg},
  \citenamefont {Vermeulen}, \citenamefont {Schouten}, \citenamefont
  {Abell{\'a}n}, \citenamefont {Amaya}, \citenamefont {Pruneri}, \citenamefont
  {Mitchell}, \citenamefont {Markham}, \citenamefont {Twitchen}, \citenamefont
  {Elkouss}, \citenamefont {Wehner}, \citenamefont {Taminiau},\ and\
  \citenamefont {Hanson}}]{Hensen2015}%
  \BibitemOpen
  \bibfield  {author} {\bibinfo {author} {\bibfnamefont {B.}~\bibnamefont
  {Hensen}}, \bibinfo {author} {\bibfnamefont {H.}~\bibnamefont {Bernien}},
  \bibinfo {author} {\bibfnamefont {A.~E.}\ \bibnamefont {Dr{\'e}au}}, \bibinfo
  {author} {\bibfnamefont {A.}~\bibnamefont {Reiserer}}, \bibinfo {author}
  {\bibfnamefont {N.}~\bibnamefont {Kalb}}, \bibinfo {author} {\bibfnamefont
  {M.~S.}\ \bibnamefont {Blok}}, \bibinfo {author} {\bibfnamefont
  {J.}~\bibnamefont {Ruitenberg}}, \bibinfo {author} {\bibfnamefont {R.~F.~L.}\
  \bibnamefont {Vermeulen}}, \bibinfo {author} {\bibfnamefont {R.~N.}\
  \bibnamefont {Schouten}}, \bibinfo {author} {\bibfnamefont {C.}~\bibnamefont
  {Abell{\'a}n}}, \bibinfo {author} {\bibfnamefont {W.}~\bibnamefont {Amaya}},
  \bibinfo {author} {\bibfnamefont {V.}~\bibnamefont {Pruneri}}, \bibinfo
  {author} {\bibfnamefont {M.~W.}\ \bibnamefont {Mitchell}}, \bibinfo {author}
  {\bibfnamefont {M.}~\bibnamefont {Markham}}, \bibinfo {author} {\bibfnamefont
  {D.~J.}\ \bibnamefont {Twitchen}}, \bibinfo {author} {\bibfnamefont
  {D.}~\bibnamefont {Elkouss}}, \bibinfo {author} {\bibfnamefont
  {S.}~\bibnamefont {Wehner}}, \bibinfo {author} {\bibfnamefont {T.~H.}\
  \bibnamefont {Taminiau}},\ and\ \bibinfo {author} {\bibfnamefont
  {R.}~\bibnamefont {Hanson}},\ }\bibfield  {title} {\bibinfo {title}
  {Loophole-free {B}ell inequality violation using electron spins separated by
  1.3 kilometres},\ }\href {https://doi.org/10.1038/nature15759} {\bibfield
  {journal} {\bibinfo  {journal} {Nature}\ }\textbf {\bibinfo {volume} {526}},\
  \bibinfo {pages} {682} (\bibinfo {year} {2015})}\BibitemShut {NoStop}%
\bibitem [{\citenamefont {Nadlinger}\ \emph {et~al.}(2022)\citenamefont
  {Nadlinger}, \citenamefont {Drmota}, \citenamefont {Nichol}, \citenamefont
  {Araneda}, \citenamefont {Main}, \citenamefont {Srinivas}, \citenamefont
  {Lucas}, \citenamefont {Ballance}, \citenamefont {Ivanov}, \citenamefont
  {Tan}, \citenamefont {Sekatski}, \citenamefont {Urbanke}, \citenamefont
  {Renner}, \citenamefont {Sangouard},\ and\ \citenamefont
  {Bancal}}]{Nadlinger_2022}%
  \BibitemOpen
  \bibfield  {author} {\bibinfo {author} {\bibfnamefont {D.~P.}\ \bibnamefont
  {Nadlinger}}, \bibinfo {author} {\bibfnamefont {P.}~\bibnamefont {Drmota}},
  \bibinfo {author} {\bibfnamefont {B.~C.}\ \bibnamefont {Nichol}}, \bibinfo
  {author} {\bibfnamefont {G.}~\bibnamefont {Araneda}}, \bibinfo {author}
  {\bibfnamefont {D.}~\bibnamefont {Main}}, \bibinfo {author} {\bibfnamefont
  {R.}~\bibnamefont {Srinivas}}, \bibinfo {author} {\bibfnamefont {D.~M.}\
  \bibnamefont {Lucas}}, \bibinfo {author} {\bibfnamefont {C.~J.}\ \bibnamefont
  {Ballance}}, \bibinfo {author} {\bibfnamefont {K.}~\bibnamefont {Ivanov}},
  \bibinfo {author} {\bibfnamefont {E.~Y.-Z.}\ \bibnamefont {Tan}}, \bibinfo
  {author} {\bibfnamefont {P.}~\bibnamefont {Sekatski}}, \bibinfo {author}
  {\bibfnamefont {R.~L.}\ \bibnamefont {Urbanke}}, \bibinfo {author}
  {\bibfnamefont {R.}~\bibnamefont {Renner}}, \bibinfo {author} {\bibfnamefont
  {N.}~\bibnamefont {Sangouard}},\ and\ \bibinfo {author} {\bibfnamefont
  {J.-D.}\ \bibnamefont {Bancal}},\ }\bibfield  {title} {\bibinfo {title}
  {Experimental quantum key distribution certified by {B}ell’s theorem},\
  }\href {https://doi.org/10.1038/s41586-022-04941-5} {\bibfield  {journal}
  {\bibinfo  {journal} {Nature}\ }\textbf {\bibinfo {volume} {607}},\ \bibinfo
  {pages} {682–686} (\bibinfo {year} {2022})}\BibitemShut {NoStop}%
\bibitem [{\citenamefont {Bhavsar}\ \emph {et~al.}(2023)\citenamefont
  {Bhavsar}, \citenamefont {Ragy},\ and\ \citenamefont
  {Colbeck}}]{bhavsar2023improved}%
  \BibitemOpen
  \bibfield  {author} {\bibinfo {author} {\bibfnamefont {R.}~\bibnamefont
  {Bhavsar}}, \bibinfo {author} {\bibfnamefont {S.}~\bibnamefont {Ragy}},\ and\
  \bibinfo {author} {\bibfnamefont {R.}~\bibnamefont {Colbeck}},\ }\bibfield
  {title} {\bibinfo {title} {Improved device-independent randomness expansion
  rates using two sided randomness},\ }\href
  {https://doi.org/10.1088/1367-2630/acf393} {\bibfield  {journal} {\bibinfo
  {journal} {New Jounal of Physics}\ }\textbf {\bibinfo {volume} {25}},\
  \bibinfo {pages} {093035} (\bibinfo {year} {2023})}\BibitemShut {NoStop}%
\bibitem [{\citenamefont {Paulsen}(2003)}]{paulsen_2003}%
  \BibitemOpen
  \bibfield  {author} {\bibinfo {author} {\bibfnamefont {V.}~\bibnamefont
  {Paulsen}},\ }\href {https://doi.org/10.1017/CBO9780511546631} {\emph
  {\bibinfo {title} {Completely Bounded Maps and Operator Algebras}}},\
  Cambridge Studies in Advanced Mathematics\ (\bibinfo  {publisher} {Cambridge
  University Press},\ \bibinfo {year} {2003})\BibitemShut {NoStop}%
\bibitem [{\citenamefont {Bamps}\ and\ \citenamefont
  {Pironio}(2015)}]{BampsPironio}%
  \BibitemOpen
  \bibfield  {author} {\bibinfo {author} {\bibfnamefont {C.}~\bibnamefont
  {Bamps}}\ and\ \bibinfo {author} {\bibfnamefont {S.}~\bibnamefont
  {Pironio}},\ }\bibfield  {title} {\bibinfo {title} {Sum-of-squares
  decompositions for a family of {C}lauser-{H}orne-{S}himony-{H}olt-like
  inequalities and their application to self-testing},\ }\href
  {https://doi.org/10.1103/PhysRevA.91.052111} {\bibfield  {journal} {\bibinfo
  {journal} {Physical Review A}\ }\textbf {\bibinfo {volume} {91}},\ \bibinfo
  {pages} {052111} (\bibinfo {year} {2015})}\BibitemShut {NoStop}%
\bibitem [{\citenamefont {Cui}\ \emph {et~al.}(2020)\citenamefont {Cui},
  \citenamefont {Mehta}, \citenamefont {Mousavi},\ and\ \citenamefont
  {Nezhadi}}]{Cui_2020}%
  \BibitemOpen
  \bibfield  {author} {\bibinfo {author} {\bibfnamefont {D.}~\bibnamefont
  {Cui}}, \bibinfo {author} {\bibfnamefont {A.}~\bibnamefont {Mehta}}, \bibinfo
  {author} {\bibfnamefont {H.}~\bibnamefont {Mousavi}},\ and\ \bibinfo {author}
  {\bibfnamefont {S.~S.}\ \bibnamefont {Nezhadi}},\ }\bibfield  {title}
  {\bibinfo {title} {A generalization of {CHSH} and the algebraic structure of
  optimal strategies},\ }\href {https://doi.org/10.22331/q-2020-10-21-346}
  {\bibfield  {journal} {\bibinfo  {journal} {{Quantum}}\ }\textbf {\bibinfo
  {volume} {4}},\ \bibinfo {pages} {346} (\bibinfo {year} {2020})}\BibitemShut
  {NoStop}%
\bibitem [{\citenamefont {Franz}\ \emph {et~al.}(2011)\citenamefont {Franz},
  \citenamefont {Furrer},\ and\ \citenamefont {Werner}}]{Franz_2011}%
  \BibitemOpen
  \bibfield  {author} {\bibinfo {author} {\bibfnamefont {T.}~\bibnamefont
  {Franz}}, \bibinfo {author} {\bibfnamefont {F.}~\bibnamefont {Furrer}},\ and\
  \bibinfo {author} {\bibfnamefont {R.~F.}\ \bibnamefont {Werner}},\ }\bibfield
   {title} {\bibinfo {title} {Extremal quantum correlations and cryptographic
  security},\ }\href {https://doi.org/10.1103/physrevlett.106.250502}
  {\bibfield  {journal} {\bibinfo  {journal} {Physical Review Letters}\
  }\textbf {\bibinfo {volume} {106}},\ \bibinfo {pages} {250502} (\bibinfo
  {year} {2011})}\BibitemShut {NoStop}%
\bibitem [{\citenamefont {Kreyszig}(1991)}]{kreyszig1991introductory}%
  \BibitemOpen
  \bibfield  {author} {\bibinfo {author} {\bibfnamefont {E.}~\bibnamefont
  {Kreyszig}},\ }\href@noop {} {\emph {\bibinfo {title} {Introductory
  functional analysis with applications}}},\ Vol.~\bibinfo {volume} {17}\
  (\bibinfo  {publisher} {John Wiley \& Sons},\ \bibinfo {year}
  {1991})\BibitemShut {NoStop}%
\bibitem [{\citenamefont {Scholz}\ and\ \citenamefont
  {Werner}(2008)}]{scholz2008}%
  \BibitemOpen
  \bibfield  {author} {\bibinfo {author} {\bibfnamefont {V.~B.}\ \bibnamefont
  {Scholz}}\ and\ \bibinfo {author} {\bibfnamefont {R.~F.}\ \bibnamefont
  {Werner}},\ }\href {https://arxiv.org/abs/0812.4305} {\bibinfo {title}
  {Tsirelson's problem}} (\bibinfo {year} {2008}),\ \Eprint
  {https://arxiv.org/abs/0812.4305} {arXiv:0812.4305 [math-ph]} \BibitemShut
  {NoStop}%
\bibitem [{\citenamefont {Navascués}\ \emph {et~al.}(2012)\citenamefont
  {Navascués}, \citenamefont {Cooney}, \citenamefont {Pérez-García},\ and\
  \citenamefont {Villanueva}}]{Navascu_s_2012}%
  \BibitemOpen
  \bibfield  {author} {\bibinfo {author} {\bibfnamefont {M.}~\bibnamefont
  {Navascués}}, \bibinfo {author} {\bibfnamefont {T.}~\bibnamefont {Cooney}},
  \bibinfo {author} {\bibfnamefont {D.}~\bibnamefont {Pérez-García}},\ and\
  \bibinfo {author} {\bibfnamefont {N.}~\bibnamefont {Villanueva}},\ }\bibfield
   {title} {\bibinfo {title} {A physical approach to {T}sirelson’s problem},\
  }\href {https://doi.org/10.1007/s10701-012-9641-0} {\bibfield  {journal}
  {\bibinfo  {journal} {Foundations of Physics}\ }\textbf {\bibinfo {volume}
  {42}},\ \bibinfo {pages} {985–995} (\bibinfo {year} {2012})}\BibitemShut
  {NoStop}%
\bibitem [{\citenamefont {Brown}\ \emph {et~al.}(2024)\citenamefont {Brown},
  \citenamefont {Fawzi},\ and\ \citenamefont
  {Fawzi}}]{Brown2024deviceindependent}%
  \BibitemOpen
  \bibfield  {author} {\bibinfo {author} {\bibfnamefont {P.}~\bibnamefont
  {Brown}}, \bibinfo {author} {\bibfnamefont {H.}~\bibnamefont {Fawzi}},\ and\
  \bibinfo {author} {\bibfnamefont {O.}~\bibnamefont {Fawzi}},\ }\bibfield
  {title} {\bibinfo {title} {Device-independent lower bounds on the conditional
  von {N}eumann entropy},\ }\href {https://doi.org/10.22331/q-2024-08-27-1445}
  {\bibfield  {journal} {\bibinfo  {journal} {{Quantum}}\ }\textbf {\bibinfo
  {volume} {8}},\ \bibinfo {pages} {1445} (\bibinfo {year} {2024})}\BibitemShut
  {NoStop}%
\bibitem [{\citenamefont {Navascu\'es}\ \emph {et~al.}(2007)\citenamefont
  {Navascu\'es}, \citenamefont {Pironio},\ and\ \citenamefont
  {Ac\'{\i}n}}]{NPA1}%
  \BibitemOpen
  \bibfield  {author} {\bibinfo {author} {\bibfnamefont {M.}~\bibnamefont
  {Navascu\'es}}, \bibinfo {author} {\bibfnamefont {S.}~\bibnamefont
  {Pironio}},\ and\ \bibinfo {author} {\bibfnamefont {A.}~\bibnamefont
  {Ac\'{\i}n}},\ }\bibfield  {title} {\bibinfo {title} {Bounding the set of
  quantum correlations},\ }\href
  {https://doi.org/10.1103/PhysRevLett.98.010401} {\bibfield  {journal}
  {\bibinfo  {journal} {Phys. Rev. Lett.}\ }\textbf {\bibinfo {volume} {98}},\
  \bibinfo {pages} {010401} (\bibinfo {year} {2007})}\BibitemShut {NoStop}%
\bibitem [{\citenamefont {Navascués}\ \emph {et~al.}(2008)\citenamefont
  {Navascués}, \citenamefont {Pironio},\ and\ \citenamefont {Acín}}]{NPA2}%
  \BibitemOpen
  \bibfield  {author} {\bibinfo {author} {\bibfnamefont {M.}~\bibnamefont
  {Navascués}}, \bibinfo {author} {\bibfnamefont {S.}~\bibnamefont
  {Pironio}},\ and\ \bibinfo {author} {\bibfnamefont {A.}~\bibnamefont
  {Acín}},\ }\bibfield  {title} {\bibinfo {title} {A convergent hierarchy of
  semidefinite programs characterizing the set of quantum correlations},\
  }\href {https://doi.org/10.1088/1367-2630/10/7/073013} {\bibfield  {journal}
  {\bibinfo  {journal} {New Journal of Physics}\ }\textbf {\bibinfo {volume}
  {10}},\ \bibinfo {pages} {073013} (\bibinfo {year} {2008})}\BibitemShut
  {NoStop}%
\bibitem [{\citenamefont {Wooltorton}\ \emph {et~al.}(2025)\citenamefont
  {Wooltorton}, \citenamefont {Brown},\ and\ \citenamefont
  {Colbeck}}]{dataset}%
  \BibitemOpen
  \bibfield  {author} {\bibinfo {author} {\bibfnamefont {L.}~\bibnamefont
  {Wooltorton}}, \bibinfo {author} {\bibfnamefont {P.}~\bibnamefont {Brown}},\
  and\ \bibinfo {author} {\bibfnamefont {R.}~\bibnamefont {Colbeck}},\ }\href
  {https://doi.org/10.5281/zenodo.17235290} {\bibinfo {title} {Dataset for
  {F}ig.\ 2}} (\bibinfo {year}
  {2025})\BibitemShut {NoStop}%
\bibitem [{\citenamefont {Le}\ \emph {et~al.}(2023)\citenamefont {Le},
  \citenamefont {Meroni}, \citenamefont {Sturmfels}, \citenamefont {Werner},\
  and\ \citenamefont {Ziegler}}]{Le23}%
  \BibitemOpen
  \bibfield  {author} {\bibinfo {author} {\bibfnamefont {T.~P.}\ \bibnamefont
  {Le}}, \bibinfo {author} {\bibfnamefont {C.}~\bibnamefont {Meroni}}, \bibinfo
  {author} {\bibfnamefont {B.}~\bibnamefont {Sturmfels}}, \bibinfo {author}
  {\bibfnamefont {R.~F.}\ \bibnamefont {Werner}},\ and\ \bibinfo {author}
  {\bibfnamefont {T.}~\bibnamefont {Ziegler}},\ }\bibfield  {title} {\bibinfo
  {title} {Quantum {C}orrelations in the {M}inimal {S}cenario},\ }\href
  {https://doi.org/10.22331/q-2023-03-16-947} {\bibfield  {journal} {\bibinfo
  {journal} {{Quantum}}\ }\textbf {\bibinfo {volume} {7}},\ \bibinfo {pages}
  {947} (\bibinfo {year} {2023})}\BibitemShut {NoStop}%
\bibitem [{\citenamefont {Barizien}\ \emph {et~al.}(2024)\citenamefont
  {Barizien}, \citenamefont {Sekatski},\ and\ \citenamefont
  {Bancal}}]{Barizien24}%
  \BibitemOpen
  \bibfield  {author} {\bibinfo {author} {\bibfnamefont {V.}~\bibnamefont
  {Barizien}}, \bibinfo {author} {\bibfnamefont {P.}~\bibnamefont {Sekatski}},\
  and\ \bibinfo {author} {\bibfnamefont {J.-D.}\ \bibnamefont {Bancal}},\
  }\bibfield  {title} {\bibinfo {title} {Custom {B}ell inequalities from formal
  sums of squares},\ }\href {https://doi.org/10.22331/q-2024-05-02-1333}
  {\bibfield  {journal} {\bibinfo  {journal} {{Quantum}}\ }\textbf {\bibinfo
  {volume} {8}},\ \bibinfo {pages} {1333} (\bibinfo {year} {2024})}\BibitemShut
  {NoStop}%
\bibitem [{\citenamefont {Wooltorton}\ \emph {et~al.}(2024)\citenamefont
  {Wooltorton}, \citenamefont {Brown},\ and\ \citenamefont {Colbeck}}]{WBC3}%
  \BibitemOpen
  \bibfield  {author} {\bibinfo {author} {\bibfnamefont {L.}~\bibnamefont
  {Wooltorton}}, \bibinfo {author} {\bibfnamefont {P.}~\bibnamefont {Brown}},\
  and\ \bibinfo {author} {\bibfnamefont {R.}~\bibnamefont {Colbeck}},\
  }\bibfield  {title} {\bibinfo {title} {Device-independent quantum key
  distribution with arbitrarily small nonlocality},\ }\href
  {https://doi.org/10.1103/PhysRevLett.132.210802} {\bibfield  {journal}
  {\bibinfo  {journal} {Phys. Rev. Lett.}\ }\textbf {\bibinfo {volume} {132}},\
  \bibinfo {pages} {210802} (\bibinfo {year} {2024})}\BibitemShut {NoStop}%
\bibitem [{\citenamefont {Farkas}(2024)}]{Farkas23}%
  \BibitemOpen
  \bibfield  {author} {\bibinfo {author} {\bibfnamefont {M.}~\bibnamefont
  {Farkas}},\ }\bibfield  {title} {\bibinfo {title} {Unbounded
  device-independent quantum key rates from arbitrarily small nonlocality},\
  }\href {https://doi.org/10.1103/PhysRevLett.132.210803} {\bibfield  {journal}
  {\bibinfo  {journal} {Phys. Rev. Lett.}\ }\textbf {\bibinfo {volume} {132}},\
  \bibinfo {pages} {210803} (\bibinfo {year} {2024})}\BibitemShut {NoStop}%
\bibitem [{\citenamefont {Pereira~Alves}\ and\ \citenamefont
  {Kaniewski}(2022)}]{PereiraAlves22}%
  \BibitemOpen
  \bibfield  {author} {\bibinfo {author} {\bibfnamefont {G.}~\bibnamefont
  {Pereira~Alves}}\ and\ \bibinfo {author} {\bibfnamefont {J.}~\bibnamefont
  {Kaniewski}},\ }\bibfield  {title} {\bibinfo {title} {Optimality of any pair
  of incompatible rank-one projective measurements for some nontrivial {B}ell
  inequality},\ }\href {https://doi.org/10.1103/PhysRevA.106.032219} {\bibfield
   {journal} {\bibinfo  {journal} {Phys. Rev. A}\ }\textbf {\bibinfo {volume}
  {106}},\ \bibinfo {pages} {032219} (\bibinfo {year} {2022})}\BibitemShut
  {NoStop}%
\end{thebibliography}
\end{document}